\documentclass[aip,jcp,amsmath,amssymb,reprint]{revtex4-1}

\usepackage{graphicx}
\usepackage{dcolumn}
\usepackage{bm}
\usepackage{color}

\usepackage[utf8]{inputenc}
\usepackage[T1]{fontenc}
\usepackage{mathptmx}
\usepackage{amsthm}
\newtheorem{theorem}{Theorem}

\begin{document}

\preprint{AIP/123-QED}

\title[]{On the effective rank of canonical polyadic decomposition of electron repulsion integrals}

\author{Aleksandra Oszmian}
\author{Micha\l\ Lesiuk}
\email{m.lesiuk@uw.edu.pl}
\affiliation{\sl University of Warsaw, Faculty of Chemistry, Pasteura 1, 02-093 Warsaw, Poland}

\date{\today}

\begin{abstract}
In this paper, we study the effective rank of the canonical polyadic decomposition applied to the electron repulsion integrals, ubiquitous in quantum chemistry. We demonstrate, both mathematically and numerically, that in general the effective rank of this decomposition cannot grow linearly as a function of the system size. Moreover, we derive a lower bound for the effective rank in the form $\propto N_{\mathrm{AO}}^2/\log_2^7 N_{\mathrm{AO}}$, where $N_{\mathrm{AO}}$ is the number of atomic orbitals in the molecule, under mild conditions imposed on the decomposition threshold $\epsilon$. As a result, while a subquadratic growth of the CPD rank is not excluded, a linear relationship between the rank and $N_{\mathrm{AO}}$ cannot hold universally. The implications of these findings for the use of the canonical polyadic format to represent electron repulsion integrals in quantum chemistry are analyzed.
\end{abstract}

\maketitle

\section{\label{sec:intro} Introduction}

Electron repulsion integrals (ERI) are at the heart of quantum chemistry as necessary components to describe the Coulomb interaction between electrons when the many-body wavefunction is expanded in terms of orbital products. When the number of atomic orbital basis functions (denoted by Greek letters) is equal to $N$, the number of ERI $(\mu\nu|\sigma\lambda)$ formally increases as $\mathcal{O}(N^4)$. It is known that this unfavorable scaling can be reduced, e.g. by screening negligible contributions on-the-fly, but in practice it is often more convenient to exploit the low-rank nature of the ERI tensor from the start by expressing it as a combination of simpler quantities. For example, density fitting (DF)~\cite{whitten73,baerends73,dunlap79,alsenoy88,vahtras93} and Cholesky decomposition (CD)~\cite{beebe77,koch03,pedersen04,folkestad19} techniques effectively bring ERI into the form $(\mu\nu|\sigma\lambda) = \sum_Q B_{\mu\nu}^Q\,B_{\sigma\lambda}^Q$, where only three-index quantities appear instead of the original four-index tensor. The use of DF/CD format reduces memory footprint and increases efficiency of many quantum-chemical methods, see for example Refs.~\onlinecite{bozkaya16,deprince13,lesiuk22,peng19,deprince14,kats07,lesiuk20,epifanovsky13,pedersen04,nottoli23}, but has the disadvantage that the pairs of indices $\mu\nu$ and $\sigma\lambda$ remain ``pinned'' together in a single tensor. It is well-known that this property prevents reduction of scaling of some operations such as construction of the exchange part of the Fock matrix. Another ERI decomposition format that is free from this problem is tensor hypercontraction~\cite{hohenstein12,parrish12} (THC) which reads $(\mu\nu|\sigma\lambda) = \sum_{RS} X_{\mu R}\,X_{\nu R}\,Z_{RS}\,X_{\sigma S}\,X_{\lambda S}$. It requires only $\mathcal{O}(N^2)$ memory to store and has been used to reduce the scaling of numerous quantum chemistry methods~\cite{hohenstein13a,hohenstein13b,pokhilko25,song16,pokhilko24,hillers25a,kokkila15,song18,jiang23,song20,hummel17,hohenstein12b,hillers25b,lee19,sacchetta25} due to the complete separation of four orbital indices, representing a huge step forward over the DF/CD. Therefore, it becomes natural to look for other types of ERI decompositions that would offer additional benefits over THC.

The format that is the main focus of this work is the canonical polyadic decomposition (CPD)~\cite{hitchcock27,kolda09}. Applied to ERI, it has the following general form:
\begin{align}
\label{eq:cpd-brutal}
    (\mu\nu|\sigma\lambda) = \sum_r^R A_{\mu r}\,B_{\nu r}\,C_{\sigma r}\,D_{\lambda r}.
\end{align}
In other fields, this format is known under alternative abbreviations such as CP~\cite{kiers00} (canonical product), PARAFAC~\cite{harshman70} or CANDECOMP~\cite{carroll70}. To the best of our knowledge, this decomposition format has been first applied to ERI by Benedikt \emph{et al.}~\cite{benedikt11}. They have shown that the factors $A_{\mu r}$, $B_{\nu r}$, $\ldots$, can be computed by alternating least squares~\cite{carroll70,harshman70,faber03,tomasi06} (ALS) or accelerated gradient procedures, and have found numerically that the effective rank of this decomposition (the parameter $R$ in the above formula) grows as $N^{1.7}-N^{2.6}$ with the size of the system $N$, depending on the basis sets composition and requested accuracy. The usefulness of the CPD approximation obtained in this way has been demonstrated in various quantum-chemical contexts, for example, in MP2 calculations~\cite{schmitz17,schmitz18} or various levels of coupled-cluster theory applied to vibrational~\cite{godtliebsen13,godtliebsen15,madsen18} and electronic structure~\cite{benedikt13,schutski17,pierce21,pierce25}, but it is impossible to list all relevant applications here. Various advances in efficient determination of the CPD of ERI and other quantities have also been reported~\cite{pierce22,pierce25b}.

Note that the format given in Eq.~(\ref{eq:cpd-brutal}) violates the symmetries of the original tensor. However, one can write an analogous CPD formula that obeys all physical symmetries of the two-electron integrals and introduces no extra symmetries. It reads:
\begin{align}
\label{eq:cpd-symm}
    (\mu\nu|\sigma\lambda) = \sum_r^R \lambda_r \left( X_{\mu r}\,X_{\nu r}\,Y_{\sigma r}\,Y_{\lambda r} + Y_{\mu r}\,Y_{\nu r}\,X_{\sigma r}\,X_{\lambda r} \right),
\end{align}
where $\lambda_r$ are real scalars and the factors $X_{\mu r}$ and $Y_{\sigma r}$ can be constrained to be column-normalized, e.g. $\sum_\mu X_{\mu r}\,X_{\mu r} = 1$, without loss of generality. However, in the literature, we have not seen this format applied to the ERI tensor yet and to stay consistent with the current applications of the CPD format, we focus solely on the general form given in Eq.~(\ref{eq:cpd-brutal}). The reason why formula~(\ref{eq:cpd-brutal}) has been dominant is likely the simplicity of the ALS optimization scheme compared to Eq.~(\ref{eq:cpd-symm}), while the lack of strict symmetry has not been a problem in practice for a sufficiently low decomposition threshold. After all, the CPD format~(\ref{eq:cpd-brutal}) can always be \emph{ad hoc} symmetrized once the optimization is finished to restore the correct symmetries without meaningfully increasing the error.

It is clear that the usefulness of the CPD decomposition~(\ref{eq:cpd-brutal}) hinges to a large extent on the rank $R$ required in practice. Ideally, if the rank increases only linearly with the system size, such CPD approximation would be massively useful in quantum chemistry and would lead to immediate scaling reduction in many basic procedures, such as the Fock matrix construction. In this work, we show, both mathematically and numerically, that the rank of this decomposition cannot grow linearly for sufficiently large molecules. We establish a lower bound for the CPD rank as a function of the system size, proving that a subquadratic rank is possible. The implications of these findings for the application of CPD in the context of ERI are also discussed. Throughout the paper, we focus on Gaussian-type orbitals (GTO) as the most common choice of the atomic orbital (AO) basis in quantum chemistry (at least for systems without translational symmetry).

\section{\label{sec:lower} Lower bound for the rank of the CPD expansion}

\subsection{\label{sec:system} Model system and its basic properties}

\begin{figure}
    \centering
    \includegraphics[width=0.8\linewidth]{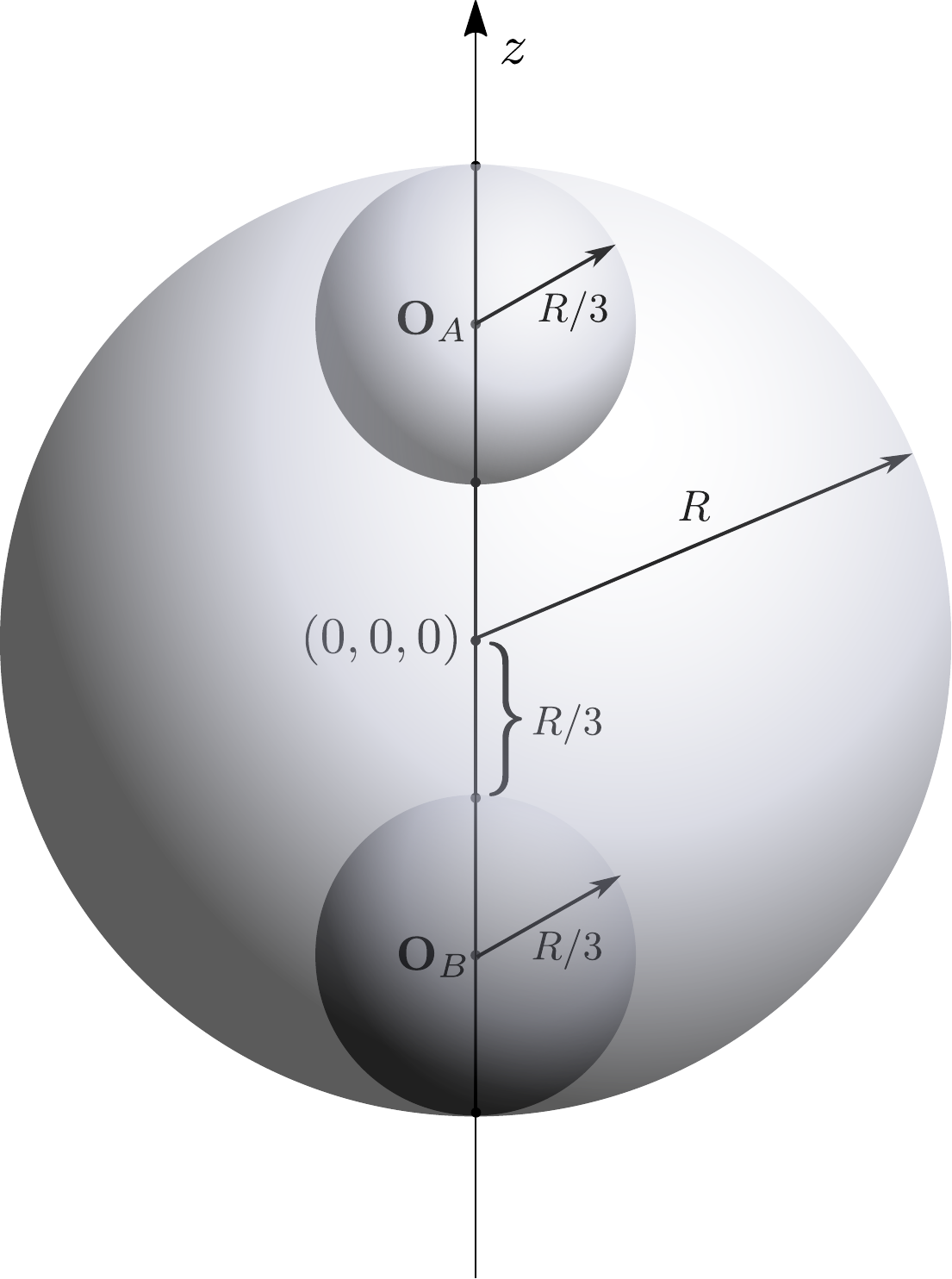}
    \caption{Graphical illustration of the construction of the model system.}
    \label{fig:system}
\end{figure}

In order to prove the main statement of the paper that will be specified shortly, we first introduce a model molecular system that will serve as the basis of our study. Consider a three-dimensional molecule/cluster that is approximately spherical in shape. The whole molecule can be then enclosed in a sphere of minimal radius $R$ such that all atoms are inside or on the surface of the sphere. Without loss of generality, we translate the molecule so that the center of the sphere lies at the origin of the coordinate system, independently of the size of the molecule. The molecule can be systematically extended and grows roughly uniformly in all spatial directions. The volume $V$ of the sphere is proportional to the number of atoms in the molecule, and for a fixed composition of the AO basis per atom it is also proportional to the total number of AO functions in the system, $N_{\mathrm{AO}}$. For convenience, let us denote the proportionality constant between the volume $V$ and $N_{\mathrm{AO}}$ by $\frac{4\pi}{3}\rho^3$, where $\rho$ is asymptotically independent of the system size, $\rho=\mathcal{O}(1)$. The parameter $\rho$ can be interpreted as a linear density of AO functions in the system per unit length. As a consequence, the radius of the sphere is equal to $R=\rho N_{\mathrm{AO}}^{1/3}$.

The tensor of all ERI within the system is denoted by the symbol $\mathbf{T}$ here and in the following. Next, we construct a specific subtensor $\mathbf{T}_{\mathrm{sub}}$ of $\mathbf{T}$ as follows. First, we draw two spheres of radius $R/3$ with centers located on the $z$ axis at the positions $\mathbf{O}_A=(0,0,2R/3)$ and $\mathbf{O}_B=(0,0,-2R/3)$. This construction is illustrated graphically in Fig.~\ref{fig:system}. We note that the choice of the $z$ axis in this procedure is arbitrary and the subsequent discussion remains valid if the spheres $A$, $B$ are simultaneously rotated around the origin of the coordinate system in any way. 

Assume that the basis set used in the calculations includes at least one uncontracted $1s$ GTOs per atom. This is true for essentially all modern basis sets of double-zeta quality or larger available in the literature~\cite{Helgaker2000}. From each atom, we pick a single uncontracted $1s$ GTOs with the lowest exponent. This choice is arbitrary and picking the largest exponent (or any other unique selection) would work as well. Denote by $\mu_A$ and $\mu_B$ all uncontracted $1s$ GTOs selected in this way located at atoms enclosed within the sphere $A$ and $B$, respectively. The subtensor $\mathbf{T}_{\mathrm{sub}}$ is composed of all ERI of the type $(\mu_A \nu_A|\sigma_B\lambda_B)$, i.e. the first orbital pair is a product of two uncontracted $1s$ GTO AO located within the sphere $A$ and analogously the second pair of uncontracted $1s$ GTO AO within the sphere $B$. Denote the number of functions $\mu_A$ and $\mu_B$ by $N_{\mathrm{sub}}^A$ and $N_{\mathrm{sub}}^B$, respectively. Clearly, both quantities are proportional to the total number of AOs in the system, i.e. $N_{\mathrm{sub}}^A \propto N_{\mathrm{AO}}$ and $N_{\mathrm{sub}}^B \propto N_{\mathrm{AO}}$, while the dimensions of the tensor $\mathbf{T}_{\mathrm{sub}}$ are $N_{\mathrm{sub}}^A\times N_{\mathrm{sub}}^A\times N_{\mathrm{sub}}^B \times N_{\mathrm{sub}}^B$. We adopt the standard convention in quantum chemistry that the AO basis set functions are normalized to unity. Additionally, we assume that AO are purely real.

\subsection{\label{sec:preliminaries} Preliminaries and definitions}

In this section, we recall several definitions and discuss error measures in ERI decompositions which are relevant in quantum chemistry. Throughout the paper, we shall consider four-dimensional tensors, where each mode of the tensor corresponds to an index of a function from the AO basis (or its subset). We recall that the members of the AO basis will be designated by Greek letters $\mu$, $\nu$, $\sigma$, $\lambda$, \ldots If not explicitly stated otherwise, the symbol $\bf T$ designates a generic tensor with elements $T_{\mu\nu\sigma\lambda}$.

The rank of the tensor $\bf T$, written $\mathrm{rank}(\mathbf{T})$, is the minimal length $R$ of the expansion in Eq.~(\ref{eq:cpd-brutal}) such that this formula holds exactly. The requirement that the expansion introduces no errors makes this definition less useful in practice. In typical applications, an error is acceptable, provided that it can be systematically controlled. To take this into account, the notion of the effective rank is usually used in the applied mathematics literature and in other fields. The effective rank of the tensor $\bf T$, written $\mathrm{rank}_\epsilon(\mathbf{T})$, is the minimum length of the expansion $R$ in Eq.~(\ref{eq:cpd-symm}) such that:
\begin{align}
    || \mathbf{T} - \bar{\mathbf{T}} || \leq \epsilon,
\end{align}
where $||*||$ denotes the tensor norm and $\bar{\mathbf{T}}$ is the CPD approximation given in Eq.~(\ref{eq:cpd-symm}). The parameter $\epsilon\geq0$ is an arbitrary non-negative real number. We can write this definition succinctly as:
\begin{align}
\label{eq:effrank}
    \mathrm{rank}_\epsilon(\mathbf{T}) = \min_R \{ R: \exists_{\bar{\mathbf{T}}} || \mathbf{T} - \bar{\mathbf{T}} || \leq \epsilon, \;\mathrm{rank}(\bar{\mathbf{T}}) = R \},
\end{align}
and clearly $\mathrm{rank}_0(\mathbf{T}) = \mathrm{rank}(\mathbf{T})$. There are different valid definitions of the tensor norm that could lead to a different effective rank. In this work, we focus on the Frobenius norm $||*||_F$ (square root of the sum of squares of all tensor elements) as the most common choice. Although different norms, in particular the spectral norm, are useful in quantum chemistry in various contexts, all existing implementations of the ALS procedure rely on the Frobenius norm to optimize the CPD of ERI and quantify the error. Furthermore, the ERI approximation error, when measured by the Frobenius norm, systematically correlates with the energy errors in MP2~\cite{benedikt11} and other electronic structure methods~\cite{benedikt13}. For instance, a threshold of $\epsilon=10^{-2}$ constrains MP2 correlation energy errors to less than one millihartree per electron pair. Tightening this threshold to $\epsilon=10^{-3}$ reduces the errors to a negligible level from the point of view of applications requiring chemical accuracy. Consequently, an $\epsilon$ between $10^{-3}$ and $10^{-2}$ is typically sufficient for practical calculations.

In the following, we will drop the subscript $F$ from the symbol of the norm and write simply $||*||$ referring to the Frobenius norm by default. We also apply this norm retroactively in the definition given in Eq.~(\ref{eq:effrank}) and the corresponding discussion.

\subsection{\label{sec:mainth} Lower bound for the effective rank of the ERI tensor}

Let $\rho=\mathcal{O}(1)$ be the linear density of the AO in the system and $R$ the radius of the sphere that encloses the system, as defined in Sec.~\ref{sec:system}. Moreover, denote by $a_{\mathrm{min}}$ the lowest exponent of the GTO present in the atomic basis of the whole system. Then we have the following theorem:

\begin{theorem}
\label{th:main} For any value of the parameter $\epsilon$ such that:
\begin{align}\label{eq:non-empty_set}
    \delta=\frac{9R^4\,e^{-4a_{\mathrm{min}}R^2/9}}{4\rho^6\sqrt{a_{\mathrm{min}} \pi}} < \epsilon < \frac{1}{\frac{3}{2R}+2R+1},
\end{align}
the effective rank of the ERI tensor $\mathbf{T}$ is bounded from below by:
\begin{align}
    \mbox{rank}_{\epsilon-\delta}(\mathbf{T}) > 
    c\,\frac{N_{\mathrm{AO}}^2}{\log_2^7 N_{\mathrm{AO}}},
\end{align}
where $c$ is a constant that is independent of the system size.
\end{theorem}
The remainder of this section is devoted to proving this theorem rigorously. As a corollary of Theorem~\ref{th:main} we note that, in general, the effective rank of the ERI tensor cannot grow linearly with the size of the system for any $R$.

A natural question here is whether the interval in Eq. (\ref{eq:non-empty_set}) is non-empty in practice. Although we do not prove this in full generality, the example provided in Sec.~\ref{sec:discussion} for a physical system confirms that the range is not only non-empty but actually quite wide.

We would like to point out that some parts of the text in this section are purely mathematical or technical in nature, and might be less interesting to readers focused on the practical applications. From their point of view, some parts of the proof can be skipped without compromising the overall understanding of the main findings. However, to understand the discussion of the results in Sec.~\ref{sec:discussion} and the numerical results given in Sec.~\ref{sec:numerical}, the readers may want to consider parts of the text dealing with the concept of a subtensor, Sec.~\ref{sec:subtensor}, and in particular with Sec.~\ref{sec:subprop} where the properties of the subtensor $\mathbf{T}_{\mathrm{sub}}$ introduced in Sec.~\ref{sec:system} are discussed. The whole proof revolves around $\mathbf{T}_{\mathrm{sub}}$ and its long-range behavior is the main culprit behind the rank explosion expressed through Theorem~\ref{th:main}.

\subsection{\label{sec:subtensor} Effective relative rank of a subtensor}

For completeness, we recall that a subtensor $\mathbf{T}_{\mathrm{sub}}$ is obtained by picking an arbitrary subset of indices in each mode of the original tensor $\mathbf{T}$ and forming $\mathbf{T}_{\mathrm{sub}}$ at the intersection of the selected indices. In the main proof discussed in the next section, we require a relation between the effective rank of a subtensor and that of the whole tensor. For the standard rank, $\mathrm{rank}(\mathbf{T})$, this relationship is straightforward, that is, the rank of the whole tensor cannot be lower than that of any of its subtensors, i.e. $\mathrm{rank}(\mathbf{T})\geq \mathrm{rank}(\mathbf{T}_{\mathrm{sub}})$ for every $\mathbf{T}_{\mathrm{sub}}$. This statement was proven in Ref.~\onlinecite{qi20}. Below we show that the same is true for the effective rank as summarized by the following theorem. 

\begin{theorem}
\label{th:subtensor}
Let $\mathbf{T}_{\mathrm{sub}}$ denote an arbitrary subtensor of $\mathbf{T}$. Then $\mathrm{rank}_\epsilon(\mathbf{T}) \geq \mathrm{rank}_\epsilon(\mathbf{T}_{\mathrm{sub}})$.
\end{theorem}

\begin{proof}
To prove this theorem we denote the effective rank (with threshold $\epsilon$) of the whole tensor $\mathbf{T}$ by $K$, i.e. $\mathrm{rank}_\epsilon(\mathbf{T}) = K$. Then there exists a tensor $\bar{\mathbf{T}}$ of rank $K$ such that:
\begin{align}
\label{eq:ineq1}
    || \mathbf{T} - \bar{\mathbf{T}} || \leq \epsilon
\end{align}
Next, we pick an arbitrary subtensor $\mathbf{T}_{\mathrm{sub}}$ and denote $\mathrm{rank}_\epsilon(\mathbf{T}_{\mathrm{sub}}) = R$. The CPD of the subtensor $\mathbf{T}_{\mathrm{sub}}$ can be obtained by restricting the indices of $\bar{\mathbf{T}}$ to the selected intersection of modes. Denote such a truncation of $\bar{\mathbf{T}}$ by $\bar{\mathbf{T}}_{\mathrm{sub}}$. As $\mathrm{rank}(\bar{\mathbf{T}})=K$, we necessarily have the upper bound $\mathrm{rank}(\bar{\mathbf{T}}_{\mathrm{sub}})\leq K$. From the properties of the norm it follows that:
\begin{align}
    || \mathbf{T}_{\mathrm{sub}} - \bar{\mathbf{T}}_{\mathrm{sub}} || \leq || \mathbf{T} - \bar{\mathbf{T}} || \leq \epsilon.
\end{align}
This shows that $\bar{\mathbf{T}}_{\mathrm{sub}}$ is a tensor with $\mathrm{rank}(\bar{\mathbf{T}}_{\mathrm{sub}})\leq K$ and an error with respect to $\mathbf{T}_{\mathrm{sub}}$ smaller than or equal to $\epsilon$. However, from the definition of effective rank in Eq.~(\ref{eq:effrank}) and the fact that $\mathrm{rank}_\epsilon(\mathbf{T}_{\mathrm{sub}}) = R$, we conclude that $R$ is the lowest possible rank of any tensor that meets this condition. Therefore, $\mathrm{rank}(\bar{\mathbf{T}}_{\mathrm{sub}})\geq R$ which also implies that $K\geq R$ and this completes the proof.
\end{proof}

\subsection{\label{sec:hadamard} Effective rank of the Hadamard product}

The Hadamard product of two tensors $\mathbf{T}_1$ and $\mathbf{T}_2$ of the same dimension in each mode, denoted by the symbol $\mathbf{T}_1\odot\mathbf{T}_2$, returns a tensor whose elements are products of the corresponding elements of the original tensors. In subsequent derivations, we will encounter some key tensors that have such a form. In anticipation of that, we have to quantify how the effective rank of a tensor being a Hadamard product is related to the effective rank of the constituting tensors. 

In the following two theorems we summarize key results related to this problem that are useful later. However, before stating these theorems, we recall two inequalities used in their proofs that are valid for the norm of the Hadamard product of tensors. The first is the usual submultiplicative property which follows from the Cauchy–Schwarz inequality applied elementwise~\cite{Horn_Johnson_1985_ch5}:
\begin{align}
\label{eq:hada1}
    || \mathbf{A} \odot \mathbf{B} || \leq || \mathbf{A} ||\cdot || \mathbf{B} ||,
\end{align}
for any tensors $\mathbf{A}$ and $\mathbf{B}$ of the same dimension in each mode. This bound is usually rather loose, and in some situations it is advantageous to apply a somewhat stronger bound:
\begin{align}
\label{eq:hada2}
    || \mathbf{A} \odot \mathbf{B} || \leq \mbox{max}(|\mathbf{A}|)\cdot || \mathbf{B} ||,
\end{align}
where the symbol $\mbox{max}(|\mathbf{A}|)$ denotes the maximum absolute element of the tensor $\mathbf{A}$, i.e. $\mbox{max}(|\mathbf{A}|)=\max_{\mu\nu\sigma\lambda} |A_{\mu\nu\sigma\lambda}|$. The proof of this statement is a straightforward generalization of the corresponding results for the norm of a matrix discussed in Ref.~\onlinecite{Horn_Johnson_1985_ch5}. Although the proofs given below rely on the properties of the Frobenius norm, the same arguments should extend to other norms, provided that they satisfy inequalities (\ref{eq:hada1}), (\ref{eq:hada2}).

\begin{theorem}
\label{th:hada1}
Let $\mathbf{T}_1$ and $\mathbf{T}_2$ be two tensors of identical dimensions in each mode. Then $\mathrm{rank}_\eta(\mathbf{T}_1\odot\mathbf{T}_2) \leq \mathrm{rank}_\epsilon(\mathbf{T}_1)\cdot\mathrm{rank}_\epsilon(\mathbf{T}_2)$, where $\eta=\big[\mbox{max}(|\mathbf{T}_1|)+\mbox{max}(|\mathbf{T}_2|)\big]\epsilon+\epsilon^2$, for any $\epsilon\geq0$.
\end{theorem}

\begin{proof}
From the statement of the theorem there exist two tensors $\bar{\mathbf{T}}_1$ and $\bar{\mathbf{T}}_2$ such that $||\mathbf{T}_1-\bar{\mathbf{T}}_1||\leq\epsilon$ and $||\mathbf{T}_2-\bar{\mathbf{T}}_2||\leq\epsilon$. Denote their ranks by $\mbox{rank}(\bar{\mathbf{T}}_1)=R_1$ and $\mbox{rank}(\bar{\mathbf{T}}_2)=R_2$.The rank bound for the Hadamard product is known for matrices\cite{Horn_Johnson_1991_ch5} and generalizes to tensors by the distributive law, giving $\mbox{rank}(\bar{\mathbf{T}}_1\odot \bar{\mathbf{T}}_2)\leq R_1R_2$. Consider the tensor $\bar{\mathbf{T}}_1\odot\bar{\mathbf{T}}_2$ for which we have:
\begin{align*}
    &||\mathbf{T}_1\odot\mathbf{T}_2 - \bar{\mathbf{T}}_1\odot\bar{\mathbf{T}}_2|| = \\
    &||\mathbf{T}_2\odot(\mathbf{T}_1-\bar{\mathbf{T}}_1) + 
    \mathbf{T}_1\odot(\mathbf{T}_2-\bar{\mathbf{T}}_2) +
    (\bar{\mathbf{T}}_1-\mathbf{T}_1)\odot(\mathbf{T}_2-\bar{\mathbf{T}}_2) || \\
    &\leq ||\mathbf{T}_2\odot(\mathbf{T}_1-\bar{\mathbf{T}}_1)|| + 
    ||\mathbf{T}_1\odot(\mathbf{T}_2-\bar{\mathbf{T}}_2) || \\
    &+ ||(\bar{\mathbf{T}}_1-\mathbf{T}_1)\odot(\mathbf{T}_2-\bar{\mathbf{T}}_2) ||,
\end{align*}
where in the last step, we have used the triangle inequality for the norm twice. Next, for the last term we apply Eq.~(\ref{eq:hada1}) which leads to:
\begin{align*}
    ||(\bar{\mathbf{T}}_1-\mathbf{T}_1)\odot(\mathbf{T}_2-\bar{\mathbf{T}}_2) || \leq
    ||(\bar{\mathbf{T}}_1-\mathbf{T}_1)||\cdot ||(\mathbf{T}_2-\bar{\mathbf{T}}_2) || \leq 
    \epsilon^2.
\end{align*}
For the first term, we apply the stronger bound defined in Eq.~(\ref{eq:hada2}) giving:
\begin{align*}
    ||\mathbf{T}_2\odot(\mathbf{T}_1-\bar{\mathbf{T}}_1)|| \leq
    \mbox{max}(|\mathbf{T}_2|)\cdot || \mathbf{T}_1-\bar{\mathbf{T}}_1 || \leq
    \mbox{max}(|\mathbf{T}_2|)\,\epsilon,
\end{align*}
and analogously for the second term. Inserting these bounds into the previous formula gives
\begin{align*}
    &||\mathbf{T}_1\odot\mathbf{T}_2 - \bar{\mathbf{T}}_1\odot\bar{\mathbf{T}}_2|| \leq 
    \mbox{max}(|\mathbf{T}_1|)\epsilon + \mbox{max}(|\mathbf{T}_2|)\epsilon + \epsilon^2.
\end{align*}
As a result, we have shown that $\bar{\mathbf{T}}_1\odot\bar{\mathbf{T}}_2$ is a tensor with rank no greater than $R_1 R_2$ which is an approximation of $\mathbf{T}_1\odot\mathbf{T}_2$ with error smaller than or equal to $\eta=\big[\mbox{max}(|\mathbf{T}_1|) + \mbox{max}(|\mathbf{T}_2|)\big]\epsilon + \epsilon^2$. This leads to the conclusion that $\mbox{rank}_\eta(\mathbf{T}_1\odot\mathbf{T}_2)\leq R_1R_2$ and this completes the proof.
\end{proof}

\begin{theorem}
\label{th:hada2}
Let $\mathbf{T}_1$ and $\mathbf{T}_2$ be two tensors of identical dimensions in each mode. Assume that $\mathbf{T}_2$ is dense, i.e. it has no zero elements, and hence its elementwise inverse $\mathbf{T}_2^{\circ -1}$ exists and has finite elements. Then $\mathrm{rank}_\epsilon(\mathbf{T}_1\odot\mathbf{T}_2^{-1}) \geq \frac{\mathrm{rank}_\eta(\mathbf{T}_1)}{\mathrm{rank}_\epsilon(\mathbf{T}_2)}$, where $\eta=\big[\mbox{max}(|\mathbf{T}_1\odot\mathbf{T}_2^{-1}|) + \mbox{max}(|\mathbf{T}_2|)\big]\epsilon+\epsilon^2$, for any $\epsilon\geq0$.
\end{theorem}

\begin{proof}
We begin by writing
\begin{align*}
    \mathbf{T}_1 = \mathbf{T}_1 \odot \mathbf{T}_2^{\circ -1} \odot \mathbf{T}_2,
\end{align*}
from which follows that for any $\eta$ we have
\begin{align*}
    \mathrm{rank}_\eta(\mathbf{T}_1) = \mathrm{rank}_\eta(\mathbf{T}_1 \odot \mathbf{T}_2^{\circ -1} \odot \mathbf{T}_2).
\end{align*}
Let us now set $\eta=\big[\mbox{max}(|\mathbf{T}_1\odot\mathbf{T}_2^{\circ -1}|) + \mbox{max}(|\mathbf{T}_2|)\big]\epsilon+\epsilon^2$. From Theorem~\ref{th:hada1} we then have the following inequality:
\begin{align*}
    \mathrm{rank}_\eta(\mathbf{T}_1 \odot \mathbf{T}_2^{\circ-1} \odot \mathbf{T}_2) \leq \mathrm{rank}_\epsilon(\mathbf{T}_1\odot\mathbf{T}_2^{\circ-1})\cdot\mathrm{rank}_\epsilon(\mathbf{T}_2),
\end{align*}
and upon combining with the previous equation we find:
\begin{align}
    \mathrm{rank}_\eta(\mathbf{T}_1) \leq \mathrm{rank}_\epsilon(\mathbf{T}_1\odot\mathbf{T}_2^{\circ-1})\cdot\mathrm{rank}_\epsilon(\mathbf{T}_2).
\end{align}
By rearranging this inequality we arrive at:
\begin{align}
    \mathrm{rank}_\epsilon(\mathbf{T}_1\odot\mathbf{T}_2^{\circ-1}) \geq \frac{\mathrm{rank}_\eta(\mathbf{T}_1)}{\mathrm{rank}_\epsilon(\mathbf{T}_2)},
\end{align}
which is exactly the statement of the proof.
\end{proof}

\subsection{\label{sec:subprop} Properties of the subtensor $\mathbf{T}_{\mathrm{sub}}$}

In previous sections, we have defined the four-dimensional subtensor $\mathbf{T}_{\mathrm{sub}}$ that will be analyzed here. In particular, our goal is to find a lower bound for the effective rank of this tensor, $\mathrm{rank}_\epsilon(\mathbf{T}_{\mathrm{sub}})$, as a function of the system size. Once this bound is known, we can invoke Theorem~\ref{th:subtensor} stating that $\mathrm{rank}_\epsilon(\mathbf{T})\geq \mathrm{rank}_\epsilon(\mathbf{T}_{\mathrm{sub}})$. As a result, the rank of the entire tensor will also be bounded from below by the same quantity as $\mathrm{rank}_\epsilon(\mathbf{T}_{\mathrm{sub}})$. The rationale for this approach is based on the fact that while the rank of the entire tensor $\mathbf{T}$ appears extremely difficult to assess, much more can be derived about the rank of the particular subtensor $\mathbf{T}_{\mathrm{sub}}$, as shown next.

Recall that the tensor $\mathbf{T}_{\mathrm{sub}}$ has elements $(\mu_A \nu_A|\sigma_B\lambda_B)$, where all GTOs are of the uncontracted $1s$ type. Such integrals can be evaluated analytically. Let us denote the positions at which the AO are located by $\mathbf{R}_\mu$ and their exponents by $a_\mu$. The integrals in question read:
\begin{align}
\label{eq:tsub}
\begin{split}
    &(\mu_A \nu_A|\sigma_B\lambda_B) = S_{\mu_A \nu_A}\,S_{\sigma_B\lambda_B} \\
    &\times \frac{\mbox{erf}\Big(\sqrt{\frac{a_{\mu_A\nu_A}\,a_{\sigma_B\lambda_B}}{a_{\mu_A\nu_A}+a_{\sigma_B\lambda_B}}} |\mathbf{R}_{\mu_A\nu_A}-\mathbf{R}_{\sigma_B\lambda_B}|\Big)}{|\mathbf{R}_{\mu_A\nu_A}-\mathbf{R}_{\sigma_B\lambda_B}|},
\end{split}
\end{align}
where $\mbox{erf}(x)$ is the error function, $a_{\mu\nu}=a_\mu + a_\nu$, $\mathbf{R}_{\mu\nu} = \frac{a_\mu \mathbf{R}_\mu + a_\nu \mathbf{R}_\nu}{a_\mu + a_\nu}$, and $S_{\mu_A \nu_A}$, $S_{\sigma_B\lambda_B}$ are the standard overlap integrals between AO in the spheres $A$ and $B$, respectively.

As the separation between the centers $\mathbf{R}_{\mu\nu}$ and $\mathbf{R}_{\sigma\lambda}$ increases, the error function in the above formula tends to unity. For a sufficiently large size of the spheres $A$ and $B$ the integrals $(\mu_A \nu_A|\sigma_B\lambda_B)$ can be approximated as:
\begin{align}
\label{eq:tmon}
    (\mu_A \nu_A|\sigma_B\lambda_B) \approx (\mu_A \nu_A|\sigma_B\lambda_B)_{\mathrm{mon}} = 
    \frac{S_{\mu_A \nu_A}\,S_{\sigma_B\lambda_B}}{|\mathbf{R}_{\mu_A\nu_A}-\mathbf{R}_{\sigma_B\lambda_B}|},
\end{align}
which is simply the interaction between monopoles (in the language of the multipole expansion). Let us denote by $\mathbf{T}_{\mathrm{sub}}^{\mathrm{(mon)}}$ the tensor with the elements $(\mu_A \nu_A|\sigma_B\lambda_B)_{\mathrm{mon}}$. Observe that according to Eq.~(\ref{eq:tmon}), $\mathbf{T}_{\mathrm{sub}}^{\mathrm{(mon)}}$ is a Hadamard product of two tensors. Let us denote the tensor in the numerator of Eq.~(\ref{eq:tmon}) with elements $S_{\mu_A \nu_A}\,S_{\sigma_B\lambda_B}$ by $\mathbf{N}$. Similarly, the tensor with elements $|\mathbf{R}_{\mu_A\nu_A}-\mathbf{R}_{\sigma_B\lambda_B}|$ in the denominator of Eq.~(\ref{eq:tmon}) will be denoted by $\mathbf{D}$.

For the purposes of the subsequent discussion, we have to find the upper bound for the norm of the difference between the tensors $\mathbf{T}_{\mathrm{sub}}^{\mathrm{(mon)}}$ and $\mathbf{T}_{\mathrm{sub}}$. First, note that the difference between $\mathbf{T}_{\mathrm{sub}}^{\mathrm{(mon)}}$ and $\mathbf{T}_{\mathrm{sub}}$ can be written as the Hadamard product:
\begin{align}
    \mathbf{T}_{\mathrm{sub}}^{\mathrm{(mon)}} - \mathbf{T}_{\mathrm{sub}} = \mathbf{E}\odot \mathbf{N},
\end{align}
where the tensor $\mathbf{N}$ has been defined above and the tensor $\mathbf{E}$ has the elements:
\begin{align}
\label{eq:etensor}
    E_{\mu\nu\sigma\lambda} = \frac{\mbox{erfc}\Big(\sqrt{\frac{a_{\mu\nu}\,a_{\sigma\lambda}}{a_{\mu\nu}+a_{\sigma\lambda}}} |\mathbf{R}_{\mu_A\nu_A}-\mathbf{R}_{\sigma_B\lambda_B}|\Big)}{|\mathbf{R}_{\mu_A\nu_A}-\mathbf{R}_{\sigma_B\lambda_B}|},
\end{align}
where we have used the standard definition $\mbox{erfc}(x)=1-\mbox{erf}(x)$. Applying the inequality given in Eq.~(\ref{eq:hada2}) we obtain:
\begin{align}
\label{eq:boundmon}
    || \mathbf{T}_{\mathrm{sub}}^{\mathrm{(mon)}} - \mathbf{T}_{\mathrm{sub}} || = 
    || \mathbf{E}\odot \mathbf{N} || \leq \mbox{max}(|\mathbf{E}|)\cdot ||\mathbf{N}||.
\end{align}
The upper bound for $||\mathbf{N}||$ can be established using the Schwarz inequality and the fact that the AO basis functions are normalized to unity:
\begin{align*}
    &||\mathbf{N}||^2 = 
    \sum_{\mu_A \nu_A\sigma_B\lambda_B}S^2_{\mu_A \nu_A}\,S^2_{\sigma_B\lambda_B} \leq \\
    &\sum_{\mu_A \nu_A\sigma_B\lambda_B} S_{\mu_A\mu_A}\,S_{\nu_A\nu_A}\,S_{\sigma_B\sigma_B}\,S_{\lambda_B\lambda_B} = \\
    &\sum_{\mu_A \nu_A\sigma_B\lambda_B} 1 = 
    (N_{\mathrm{sub}}^A)^2 \cdot (N_{\mathrm{sub}}^B)^2 \leq N_{\mathrm{AO}}^4,
\end{align*}
hence $||\mathbf{N}||\leq N_{\mathrm{AO}}^2$. This bound is obviously very loose but it is sufficient for the present purposes. It is also useful to rewrite this inequality explicitly using the dimension of the system, $R=\rho N_{\mathrm{AO}}^{1/3}$, giving $||\mathbf{N}||\leq R^6/\rho^6$.

Next, we consider the maximum element $\mbox{max}(|\mathbf{E}|)$ of the tensor $\mathbf{E}$. The function $\mbox{erfc}(x)$ is decreasing in its whole domain, so to find its maximum value we have to determine the minimum value of the argument $\sqrt{\frac{a_{\mu\nu}\,a_{\sigma\lambda}}{a_{\mu\nu}+a_{\sigma\lambda}}} |\mathbf{R}_{\mu_A\nu_A}-\mathbf{R}_{\sigma_B\lambda_B}|$. Let us denote the lowest exponent of the $1s$ function in the basis by $a_{\mathrm{min}}$ such that $a_\mu\geq a_{\mathrm{min}}$ for every $\mu$. From this follows $a_{\mu\nu} \geq 2a_{\mathrm{min}}$ and $\frac{a_{\mu\nu}\,a_{\sigma\lambda}}{a_{\mu\nu}+a_{\sigma\lambda}}\geq a_{\mathrm{min}}$. Simultaneously, the quantity $|\mathbf{R}_{\mu_A\nu_A}-\mathbf{R}_{\sigma_B\lambda_B}|$ is also bounded from below, because the points $\mathbf{R}_{\mu\nu}$ and $\mathbf{R}_{\sigma\lambda}$ are located within the two non-overlapping spheres $A$ and $B$. The closest distance between two points located anywhere within or on the surface of the spheres is $2R/3$ which leads to the bound $|\mathbf{R}_{\mu_A\nu_A}-\mathbf{R}_{\sigma_B\lambda_B}|\geq 2R/3$. Putting this together, we find:
\begin{align}
    \sqrt{\frac{a_{\mu\nu}\,a_{\sigma\lambda}}{a_{\mu\nu}+a_{\sigma\lambda}}} |\mathbf{R}_{\mu_A\nu_A}-\mathbf{R}_{\sigma_B\lambda_B}| \geq \frac{2R\sqrt{a_{\mathrm{min}}}}{3}.
\end{align}
Additionally, employing the upper bound for the complementary error function derived in Ref.~\onlinecite{karagiannidis07}, namely $\mbox{erfc}(x)\leq\frac{e^{-x^2}}{x\sqrt{\pi}}$ for $x>0$, we find in total:
\begin{align}
    \mbox{max}(|\mathbf{E}|) \leq 
    \frac{9e^{-4a_{\mathrm{min}}R^2/9}}{4R^2\sqrt{a_{\mathrm{min}} \pi}},
\end{align}
where we have taken care of the denominator in Eq.~(\ref{eq:etensor}) using the inequality $|\mathbf{R}_{\mu_A\nu_A}-\mathbf{R}_{\sigma_B\lambda_B}|\geq 2R/3$, see above. Combining this result with the bound for $||\mathbf{N}||$ and Eq.~(\ref{eq:boundmon}) we finally obtain:
\begin{align}
    || \mathbf{T}_{\mathrm{sub}}^{\mathrm{(mon)}} - \mathbf{T}_{\mathrm{sub}} || \leq
    \frac{9R^4\,e^{-4a_{\mathrm{min}}R^2/9}}{4\rho^6\sqrt{a_{\mathrm{min}} \pi}} = \delta,
\end{align}
where we have used the symbol $\delta$ as a shorthand for the right side of the inequality. Critically, the parameter $\delta$ vanishes exponentially as a function of $R$.

Now, assume that we have obtained a CPD approximation $\bar{\mathbf{T}}_{\mathrm{sub}}$ of $\mathbf{T}_{\mathrm{sub}}$ such that for some $\epsilon>0$ we have $\mbox{rank}_\epsilon(\mathbf{T}_{\mathrm{sub}})=\mbox{rank}(\bar{\mathbf{T}}_{\mathrm{sub}}) = R$. Consider the norm of the difference $||\mathbf{T}_{\mathrm{sub}}^{\mathrm{(mon)}}-\bar{\mathbf{T}}_{\mathrm{sub}}||$. From the triangle inequality, we get:
\begin{align}
\begin{split}
    &||\mathbf{T}_{\mathrm{sub}}^{\mathrm{(mon)}} - \bar{\mathbf{T}}_{\mathrm{sub}}|| =
    ||\mathbf{T}_{\mathrm{sub}}^{\mathrm{(mon)}} - \mathbf{T}_{\mathrm{sub}} +
    \mathbf{T}_{\mathrm{sub}} - \bar{\mathbf{T}}_{\mathrm{sub}}|| \\
    &\leq ||\mathbf{T}_{\mathrm{sub}}^{\mathrm{(mon)}} - \mathbf{T}_{\mathrm{sub}} ||
    + ||\mathbf{T}_{\mathrm{sub}} - \bar{\mathbf{T}}_{\mathrm{sub}}||.
\end{split}
\end{align}
According to the above discussion, the first term is bounded by $\delta$ and the second by $\epsilon$. This leads to:
\begin{align}
    ||\mathbf{T}_{\mathrm{sub}}^{\mathrm{(mon)}} - \bar{\mathbf{T}}_{\mathrm{sub}}|| 
    \leq \epsilon + \delta,
\end{align}
showing that $\bar{\mathbf{T}}_{\mathrm{sub}}$ is also an approximation to $\mathbf{T}_{\mathrm{sub}}^{\mathrm{(mon)}}$ with error no larger than $\epsilon + \delta$. This leads to the conclusion that the rank of $\mathbf{T}_{\mathrm{sub}}^{\mathrm{(mon)}}$ is necessarily bounded:
\begin{align}
    \mbox{rank}_{\epsilon+\delta}(\mathbf{T}_{\mathrm{sub}}^{\mathrm{(mon)}}) \leq
    R = \mbox{rank}(\bar{\mathbf{T}}_{\mathrm{sub}}) = \mbox{rank}_\epsilon(\mathbf{T}_{\mathrm{sub}}).
\end{align}
A useful formula is obtained by rewriting this inequality as:
\begin{align}
    \mbox{rank}_{\epsilon-\delta}(\mathbf{T}_{\mathrm{sub}}) \geq 
    \mbox{rank}_{\epsilon}(\mathbf{T}_{\mathrm{sub}}^{\mathrm{(mon)}}).
\end{align}
which holds as long as $\epsilon>\delta$. This is a manifestation of a stronger condition that the effective rank is continuous in the sense that tensors that are sufficiently close to each other have the same rank. As a side note, the same is not true for the mathematical rank -- there are known counterexamples of sequences of tensors of rank $2$ that tend to a tensor of rank $3$ in the limit of the sequence~\cite{de08,paatero00}. Using the above formula, we can now focus on the effective rank of the tensor $\mathbf{T}_{\mathrm{sub}}^{\mathrm{(mon)}}$.

Let us now consider the effective rank of the tensor $\mathbf{T}_{\mathrm{sub}}^{\mathrm{(mon)}}$. According to the above discussion, this tensor can be written as the Hadamard product $\mathbf{T}_{\mathrm{sub}}^{\mathrm{(mon)}} = \mathbf{N}\odot \mathbf{D}^{\circ-1}$. As the tensor $\mathbf{D}$ is dense, i.e. it contains no zero elements, we can apply Theorem~\ref{th:hada2} to this form, giving:
\begin{align}
    \mathrm{rank}_\epsilon(\mathbf{N}\odot \mathbf{D}^{\circ-1}) \geq \frac{\mathrm{rank}_\eta(\mathbf{N})}{\mathrm{rank}_\epsilon(\mathbf{D})},
\end{align}
where $\eta=\big[\mbox{max}(|\mathbf{N}\odot \mathbf{D}^{\circ-1}|)+\mbox{max}(|\mathbf{D}|)\big]\epsilon+\epsilon^2$. In summary, the main result of this section is the following chain of inequalities:
\begin{align}
    \mbox{rank}_{\epsilon-\delta}(\mathbf{T}) \geq 
    \mbox{rank}_{\epsilon-\delta}(\mathbf{T}_{\mathrm{sub}}) \geq 
    \mbox{rank}_{\epsilon}(\mathbf{T}_{\mathrm{sub}}^{\mathrm{(mon)}}) \geq
    \frac{\mathrm{rank}_\eta(\mathbf{N})}{\mathrm{rank}_\epsilon(\mathbf{D})}.
\end{align}
This allows us to focus on the effective ranks of the tensors $\mathbf{N}$ and $\mathbf{D}$ further in the text. Once the bounds are known for them, we can immediately infer the implications for the full tensor $\mathbf{T}$ of ERI.

\subsection{\label{sec:tensorn} Effective rank of the tensor $\mathbf{N}$}

In this section, we consider the effective rank of the tensor $\mathbf{N}$ with elements $N_{\mu_A \nu_A\sigma_B\lambda_B} = S_{\mu_A \nu_A}\,S_{\sigma_B\lambda_B}$. First, we note that for any threshold $\epsilon$, the effective rank of the tensor $\mathbf{N}$ cannot be lower than the effective rank of the matrix $\mathbf{M}$ defined as $M_{\mu_A\sigma_B,\nu_A\lambda_B} = S_{\mu_A \nu_A}\,S_{\sigma_B\lambda_B}$ with compound indices $\mu_A\sigma_B$ and $\nu_A\lambda_B$. To demonstrate this, let us assume that the effective rank of the matrix $\mathbf{M}$ is equal to $K$, i.e. $\mathrm{rank}_\epsilon(\mathbf{M})=K$. This means that $K$ is the lowest possible length of the expansion:
\begin{align}
    \bar{M}_{\mu_A\sigma_B,\nu_A\lambda_B} = \sum_s^K U_{\mu_A\sigma_B}^s\,V_{\nu_A\lambda_B}^s,
\end{align}
that fulfills the condition $||\mathbf{M}-\bar{\mathbf{M}}||\leq\epsilon$. On the other hand, denote by $R$ the effective rank of the tensor $\mathbf{N}$ with the threshold $\epsilon$. This means that there exists a CPD approximation $\bar{\mathbf{N}}$ in the form:
\begin{align}
    \bar{N}_{\mu_A \nu_A\sigma_B\lambda_B} = \sum_r^R A_{\mu_A r}\,B_{\nu_A r}\,C_{\sigma_B r} \,D_{\lambda_B r},
\end{align}
such that $||\mathbf{N}-\bar{\mathbf{N}}||\leq\epsilon$. This CPD approximation can be rearranged without introducing any extra approximations as:
\begin{align}
    \bar{N}_{\mu_A \nu_A\sigma_B\lambda_B} = \sum_r^R \big( A_{\mu_A r}\,C_{\sigma_B r} \big) \big( B_{\nu_A r}\,D_{\lambda_B r} \big) = 
    \sum_r^R \bar{U}_{\mu_A\sigma_B}^r\,\bar{V}_{\nu_A\lambda_B}^r,
\end{align}
which shows that this CPD is simultaneously an approximation of the matrix $\mathbf{M}$ with the same error bound ($\leq\epsilon$). However, from the definition of the matrix rank, $K$ is the lowest possible expansion length that satisfies this condition. This leads to the conclusion that $R\geq K$ which is the statement we wanted to prove, $\mbox{rank}_\epsilon(\mathbf{N})\geq \mbox{rank}_\epsilon(\mathbf{M})$.

Next, we consider the effective rank of the matrix $\mathbf{M}$. First, note that this matrix is simply the Kronecker product of the matrices $\mathbf{S}_A$ and $\mathbf{S}_B$ with elements $S_{\mu_A \nu_A}$ and $S_{\sigma_B\lambda_B}$, respectively, i.e. $\mathbf{M}=\mathbf{S}_A\otimes \mathbf{S}_B$. It is known that the eigenvalues of a matrix that is a Kronecker product are simply products of eigenvalues of the constituting matrices. Denote the eigenvalues of $\mathbf{S}_A$ and $\mathbf{S}_B$ by $e_r^A$ and $e_r^B$, respectively. Without loss of generality, we assume that they are arranged in a non-increasing order, i.e. $e_1^A\geq e_2^A\geq \ldots$, and similarly for the matrix $B$. The eigenvalues of $\mathbf{M}$ are then simply all possible products of $e_r^A\,e_s^B$. By Eckart-Young-Mirsky theorem~\cite{Eckart_Young_1936}, the best approximation to a matrix of a given rank is obtained by truncating the eigenvalue decomposition, i.e. dropping the lowest eigenvalues until the desired rank is achieved. Our goal is to establish lower and upper bounds for the number of products $e_r^A e_s^B$ that enter this expansion for a given $0<\epsilon<1$, i.e. for $\mbox{rank}_\epsilon(\mathbf{M})$.

Let us begin with the lower bound. Consider the effective ranks $\mathbf{S}_A$ and $\mathbf{S}_B$ with the threshold $\sqrt{\epsilon}$, i.e. $\mbox{rank}_{\sqrt{\epsilon}}(\mathbf{S}_A) = R_A$ and $\mbox{rank}_{\sqrt{\epsilon}}(\mathbf{S}_B) = R_B$. Note that the square root of $\epsilon$ is used in these definitions as this is critical for the subsequent arguments. From the definition of the rank, we have:
\begin{align}
\label{eq:frob1}
    \sum_{r=R_A+1}^{N_A} (e_r^A)^2 \leq \epsilon, \;\;\;\mbox{and}\;\;\;
    \sum_{r=R_A}^{N_A} (e_r^A)^2 > \epsilon,
\end{align}
and similarly for the matrix $\mathbf{S}_B$. Let us now consider the set of all possible products of $e_r^A$ and $e_s^B$. We divide it into four subsets:
\begin{align*}
   &\Omega_{\mathrm{hh}} = \{ e_r^A\,e_s^B: r\in[1,R_A], s\in[1,R_B], \mbox{without\;} r=R_A, s=R_B \}, \\
   &\Omega_{\mathrm{ht}} = \{ e_r^A\,e_s^B: r\in[1,R_A-1], s\in[R_B+1,N_B] \}, \\
   &\Omega_{\mathrm{th}} = \{ e_r^A\,e_s^B: r\in[R_A+1,N_A], s\in[1,R_B-1] \}, \\
   &\Omega_{\mathrm{tt}} = \{ e_r^A\,e_s^B: r\in[R_A,N_A], s\in[R_B,N_B] \}, \\
   &\Omega_{\mathrm{all}} = \Omega_{\mathrm{hh}} \cup \Omega_{\mathrm{ht}} \cup \Omega_{\mathrm{th}} \cup \Omega_{\mathrm{tt}}
\end{align*}
where the subscripts ``h'' and ``t'' stand for head and tail. This partition is illustrated graphically in Fig.~\ref{fig:m-matrix}. The four sets defined above together cover all possible products of $e_r^A$ and $e_s^B$, i.e. $\Omega_{\mathrm{all}}$.

\begin{figure}
    \centering
    \includegraphics[width=0.60\linewidth]{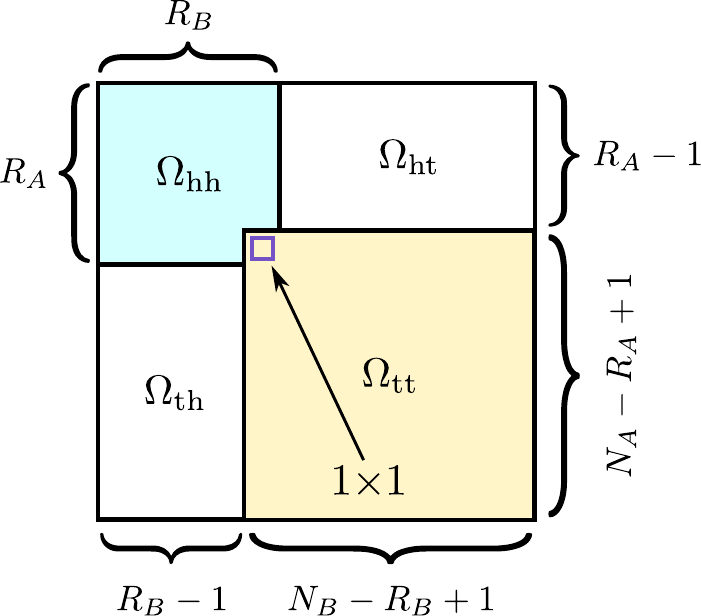}
    \caption{Illustration of the partitioning of the eigenvalues of the $\mathbf{M}$ matrix.}
    \label{fig:m-matrix}
\end{figure}

Consider a set $\Omega_{\mathrm{max}}$ that contains $R_AR_B - 1$ the largest elements of the set $\Omega_{\mathrm{all}}$. We argue that the set $\Omega_{\mathrm{max}}$ is a subset of the union of sets $\Omega_{\mathrm{hh}} \cup \Omega_{\mathrm{ht}} \cup \Omega_{\mathrm{th}}$, i.e. the set $\Omega_{\mathrm{min}} = \Omega_{\mathrm{all}} \setminus \Omega_{\mathrm{max}}$ always contains the entire set $\Omega_{\mathrm{tt}}$. Indeed, note first that the cardinality of $\Omega_{\mathrm{min}}$ is equal to $N_AN_B - R_AR_B + 1$ and is always greater than the cardinality of $\Omega_{\mathrm{tt}}$, since we consider $R_A, R_B > 1$. Second, consider the maximum element $\mathrm{max}\{ \Omega_{\mathrm{tt}} \} = e_{R_A}^A e_{R_B}^B$. Clearly, if $e_i^A e_j^B \geq e_{R_A}^A e_{R_B}^B$, then $i \leq R_A$ or $j \leq R_B$ and thus $e_i^A e_j^B \in \Omega_{\mathrm{all}} \setminus \Omega_{\mathrm{tt}}$, because the products $e_i^A e_j^B$ do not increase along each row and column of the matrix illustrated in Fig.~\ref{fig:m-matrix}. Therefore, if we wanted to obtain the best rank $R_A R_B - 1$ approximation to the matrix $\mathbf{M}$, we would pick only the products of eigenvalues from a subset of $\Omega_{\mathrm{hh}} \cup \Omega_{\mathrm{ht}} \cup \Omega_{\mathrm{th}}$. We will show that the error of an approximation obtained in this way is larger than $\epsilon$. To this end consider rank $R_A R_B - 1$ approximation $\bar{\mathbf{M}}$ to $\mathbf{M}$ obtained by retaining only the eigenvalues from $\Omega_{\mathrm{max}}$. The square of the norm of the error, i.e. $||\mathbf{M} - \bar{\mathbf{M}}||^2$ is given by:
\begin{align}
        ||\mathbf{M} - \bar{\mathbf{M}}||^2 = 
    \sum_{rs\in \Omega_{\mathrm{all}} \setminus \Omega_{\mathrm{max}}} \big( e_r^Ae_s^B\big)^2 = 
    \sum_{rs\in \Omega_{\mathrm{min}} } \big( e_r^Ae_s^B\big)^2,
\end{align}
which can be rewritten exposing the set $\Omega_{\mathrm{tt}}$ (which is a subset of $\Omega_{\mathrm{min}}$ according to the above discussion):
\begin{align}
\label{eq:frob2}
||\mathbf{M} - \bar{\mathbf{M}}||^2 = 
    \sum_{rs\in \Omega_{\mathrm{min}} \setminus \Omega_{\mathrm{tt}}} \big( e_r^Ae_s^B\big)^2 +
    \sum_{rs\in \Omega_{\mathrm{tt}}} \big( e_r^Ae_s^B\big)^2.
\end{align}
The last term can be written more explicitly as:
\begin{align}
    \sum_{rs\in \Omega_{\mathrm{tt}}} \big( e_r^Ae_s^B\big)^2 = 
    \sum_{r=R_A}^{N_A} (e_r^A)^2 \sum_{r=R_B}^{N_B} (e_r^B)^2 > \epsilon^2,
\end{align}
by recalling the definition of $\Omega_{\mathrm{tt}}$ and using Eq.~(\ref{eq:frob1}). As the remaining terms in Eq.~(\ref{eq:frob2}) are non-negative, this immediately gives us $||\mathbf{M} - \bar{\mathbf{M}}||^2>\epsilon^2$ or $||\mathbf{M} - \bar{\mathbf{M}}||>\epsilon$. This shows that the error of the best rank $R_A R_B - 1$ approximation is larger than $\epsilon$ and thus $\mbox{rank}_\epsilon(\mathbf{M})>R_A R_B - 1$. As a consequence, $\mbox{rank}_\epsilon(\mathbf{M})\geq R_A R_B$, and by recalling the definitions of $R_A$ and $R_B$ we find the lower bound:
\begin{align}
\label{eq:mbound1}
    \mbox{rank}_\epsilon(\mathbf{M}) \geq \mbox{rank}_{\sqrt{\epsilon}}(\mathbf{S}_A)\cdot\mbox{rank}_{\sqrt{\epsilon}}(\mathbf{S}_B).
\end{align}
The upper bound for $\mbox{rank}_\epsilon(\mathbf{M})$ is obtained trivially as:
\begin{align}
    \mbox{rank}_\epsilon(\mathbf{M}) \leq \mbox{rank}_0(\mathbf{M}) \leq 
    N_{\mathrm{AO}}^2,
\end{align}
because the rank of a matrix $\mathbf{M}$ cannot exceed its dimension which is, by construction, necessarily no greater than $N_{\mathrm{AO}}^2$. In total, we have the following bounds:
\begin{align}
\label{eq:mbound2}
    \mbox{rank}_{\sqrt{\epsilon}}(\mathbf{S}_A)\cdot\mbox{rank}_{\sqrt{\epsilon}}(\mathbf{S}_B)
    \leq \mbox{rank}_\epsilon(\mathbf{M}) \leq N_{\mathrm{AO}}^2.
\end{align}
In the final step, we employ the following theorem:
\begin{theorem}
\label{th:overlap}
Let $\mathbf{S}$ be the overlap matrix of dimension $N\times N$ between an arbitrary set of AO (normalized to unity) in the system described in Sec.~\ref{sec:system}. Then for any $0<\epsilon<1$, $\mbox{rank}_\epsilon(\mathbf{S})>c_\epsilon N$, where the parameter $c_\epsilon>0$ depends only on $\epsilon$ and the composition of the AO basis set per atom, but is independent of $N$.
\end{theorem}
The proof of this theorem is somewhat technical and has been included in Appendix~\ref{sec:appa} rather than in the main text. 

To take advantage of Theorem~\ref{th:overlap}, we assume that $\epsilon<1$ and hence $\sqrt{\epsilon}<1$. Then clearly $\mbox{rank}_{\sqrt{\epsilon}}(\mathbf{S}_A)>c_{\sqrt{\epsilon}} \,N_{\mathrm{sub}}^A$ (and analogously for the second matrix). By construction, $N_{\mathrm{sub}}^A\propto N_{\mathrm{AO}}$, and hence there exists a constant $c_\epsilon^{(2)}$ such that:
\begin{align}
    \mbox{rank}_\epsilon(\mathbf{M}) > c_\epsilon^{(2)} N_{\mathrm{AO}}^2.
\end{align}
This shows that the CPD rank of the tensor $\mathbf{N}$ also increases at least quadratically with the system size, i.e. $\mbox{rank}_\epsilon(\mathbf{N}) > c_\epsilon^{(2)} N_{\mathrm{AO}}^2$, for any $0<\epsilon<1$, which is the main statement we wanted to prove in this section.

\subsection{\label{sec:dhada} Tensor $D$ as a Hadamard product}

Our next major goal is to establish an upper bound for the effective rank of the tensor $\mathbf{D}$. However, this problem is difficult to attack directly. To simplify it, we split this task into two components. First, we will consider the rank of the tensor $\mathbf{G} = \mathbf{D}\odot \mathbf{D}$ and subsequently focus on the tensor $\mathbf{D}^{\circ-1}$. To show how this translates into the effective rank of the target tensor $\mathbf{D}$, we start with the obvious identity $\mathbf{D} = \mathbf{D}\odot \mathbf{D} \odot \mathbf{D}^{\circ-1} = \mathbf{G} \odot \mathbf{D}^{\circ-1}$. Let us assume that we know the exact (mathematical) rank of the tensor $\mathbf{G}$, i.e. $\mbox{rank}(\mathbf{G})$. The expression for $\mbox{rank}(\mathbf{G})$ will be derived in the next section. This means that there exists a tensor $\bar{\mathbf{G}}$ in the CPD format~(\ref{eq:cpd-brutal}) such that $\mbox{rank}(\bar{\mathbf{G}}) = \mbox{rank}(\mathbf{G})$ and $\bar{\mathbf{G}}=\mathbf{G}$. Moreover, assume that we know the expression for the effective rank of the tensor $\mathbf{D}^{\circ-1}$ with the threshold $0<\epsilon<1$, i.e. $\mbox{rank}_\epsilon(\mathbf{D}^{\circ-1})$. As a result, there exists a tensor $\bar{\mathbf{D}}^{\circ-1}$ in the CPD form~(\ref{eq:cpd-brutal}) with rank $\mbox{rank}_\epsilon(\mathbf{D}^{\circ-1})$ such that $||\mathbf{D}^{\circ-1} - \bar{\mathbf{D}}^{\circ -1} || < \epsilon$. Consider now the Hadamard product $\mathbf{G}\odot \bar{\mathbf{D}}^{\circ -1}$ which can also be written in the CPD format. The Frobenius error of the approximation of $D$ reads:
\begin{align}
\begin{split}
    &|| \mathbf{D} - \mathbf{G}\odot \bar{\mathbf{D}}^{\circ -1} || =
    || \mathbf{G} \odot \mathbf{D}^{\circ -1} - \mathbf{G}\odot \bar{\mathbf{D}}^{\circ -1} || = \\
    &|| \mathbf{G} \odot ( \mathbf{D}^{\circ -1} - \bar{\mathbf{D}}^{\circ -1} ) ||,
\end{split}
\end{align}
and by employing the inequality from Eq.~(\ref{eq:hada2}) we find:
\begin{align*}
    || \mathbf{D} - \mathbf{G}\odot \bar{\mathbf{D}}^{\circ -1} || \leq
    \mbox{max}(|\mathbf{G}|)\cdot || \mathbf{D}^{\circ -1} - \bar{\mathbf{D}}^{\circ -1} || \leq
    \mbox{max}(|\mathbf{G}|)\cdot \epsilon.
\end{align*}
The maximum element of the tensor $\mathbf{G}$ is smaller than $(2R)^2$, because the distance between any pair of points cannot exceed the total diameter ($2R$) of the entire system. This leads to the error bound:
\begin{align*}
    || \mathbf{D} - \mathbf{G}\odot \bar{\mathbf{D}}^{\circ -1} || \leq 
    (2R)^2\cdot \epsilon.
\end{align*}
This shows that the rank of the tensor $\mathbf{D}$ with threshold $(2R)^2\cdot \epsilon$ is bounded from above by:
\begin{align}
    \mbox{rank}_{(2R)^2\cdot \epsilon}(\mathbf{D}) \leq \mbox{rank}(\mathbf{G}) 
    \cdot \mbox{rank}_\epsilon(\mathbf{D}^{\circ -1}).
\end{align}
An equivalent form of this inequality that will directly be used later reads:
\begin{align}
\label{eq:rankd1}
    \mbox{rank}_{\epsilon}(\mathbf{D}) \leq \mbox{rank}(\mathbf{G}) 
    \cdot \mbox{rank}_{\epsilon/(2R)^2}(\mathbf{D}^{\circ -1}),
\end{align}
which is obtained simply by redefining $\epsilon$. Therefore, to find the upper bound for the tensor $\mathbf{D}$ we have to find the mathematical rank of the tensor $\mathbf{G}$, and an upper bound for the effective rank of the tensor $\mathbf{D}^{\circ -1}$ with a modified threshold $\epsilon/(2R)^2$. These are the goals of the next two sections.

\subsection{\label{sec:tensord2} Rank of the tensor $\mathbf{G} = \mathbf{D}\odot \mathbf{D}$}

Written elementwise, the tensor $\mathbf{G}$ is given by:
\begin{align}
    G_{\mu_A\nu_A\sigma_B\lambda_B} = |\mathbf{R}_{\mu_A\nu_A}-\mathbf{R}_{\sigma_B\lambda_B}|^2,
\end{align}
and can be expanded to the form:
\begin{align}
\begin{split}
    G_{\mu_A\nu_A\sigma_B\lambda_B} &= 
    |\mathbf{R}_{\mu_A\nu_A}|^2
    -2 X_{\mu_A\nu_A}\,X_{\sigma_B\lambda_B}
    -2 Y_{\mu_A\nu_A}\,Y_{\sigma_B\lambda_B} \\
    &-2 Z_{\mu_A\nu_A}\,Z_{\sigma_B\lambda_B}
    +|\mathbf{R}_{\sigma_B\lambda_B}|^2,
\end{split}
\end{align}
where we have used $\mathbf{R}_{\mu_A\nu_A}=[X_{\mu_A\nu_A},Y_{\mu_A\nu_A},Z_{\mu_A\nu_A}]$ and similarly for $\mathbf{R}_{\sigma_B\lambda_B}$. Let us denote by $\mathbf{M}^A$ a matrix with elements $M_{\mu_A\nu_A}^A = |\mathbf{R}_{\mu_A\nu_A}|^2$ and by $\mathbf{X}^A$ a matrix with elements $X_{\mu_A\nu_A}^A = X_{\mu_A\nu_A}$. Similar definitions hold for the counterparts of these quantities from the sphere $B$. In this notation, the tensor $\mathbf{G}$ can be rewritten in the following form:
\begin{align}
\begin{split}
    \mathbf{G} = \mathbf{M}^A\otimes\mathbf{J}^B - 2\cdot\mathbf{X}^A\otimes\mathbf{X}^B -2\cdot\mathbf{Y}^A\otimes\mathbf{Y}^B \\
-2\cdot\mathbf{Z}^A\otimes\mathbf{Z}^B + \mathbf{J}^A\otimes\mathbf{M}^B,
\end{split}
\end{align}
where the matrices $\mathbf{J}^A$ and $\mathbf{J}^B$ have the same dimensions as $\mathbf{M}^A$ and $\mathbf{M}^B$, respectively, but all of their elements are equal to unity. Applying the standard relationships for the rank of a sum and the Kronecker product of matrices~\cite{meyer2000} we arrive at:\begin{align}
    &\mbox{rank}(\mathbf{A}\otimes \mathbf{B}) = \mbox{rank}(\mathbf{A})\cdot\mbox{rank}(\mathbf{B}),\\
    &\mbox{rank}(\mathbf{A} + \mathbf{B}) \leq \mbox{rank}(\mathbf{A}) + \mbox{rank}(\mathbf{B}).
\end{align}
Taking into account that the rank of the three Cartesian components of $\mathbf{R}_{\mu_A\nu_A}$ is not necessarily identical, we obtain the upper bound for the rank of tensor $\mathbf{G}$:
\begin{align}
\label{eq:grank1}
\begin{split}
    &\mbox{rank}(\mathbf{G}) \leq \mbox{rank}(\mathbf{M}^A)
    + \mbox{rank}(\mathbf{X}^A)\,\mbox{rank}(\mathbf{X}^B) \\
    &+ \mbox{rank}(\mathbf{Y}^A)\,\mbox{rank}(\mathbf{Y}^B) 
    + \mbox{rank}(\mathbf{Z}^A)\,\mbox{rank}(\mathbf{Z}^B) \\
    &+ \mbox{rank}(\mathbf{M}^B).
\end{split}
\end{align}
On the right-hand side of the above inequality, the symbol $\mbox{rank}(*)$ stands for the standard matrix rank, i.e. the number of linearly independent rows/columns. This inequality is a direct consequence of the matrix rank invariance under multiplication by a non-zero scalar.

Our next goal is to establish an upper bound for the rank of the matrices that appear on the right-hand side of Eq.~(\ref{eq:grank1}). For simplicity, we will focus on the elements of the matrix $\mathbf{X}^A$ here, but identical formulas remain valid also for $\mathbf{Y}^A$ and $\mathbf{Z}^A$ if the Cartesian coordinates are exchanged throughout.

Consider first the rank of the matrix $\mathbf{X}^A$ with elements explicitly given by:
\begin{align}
    X_{\mu_A\nu_A} = \frac{a_{\mu_A} X_{\mu_A} + a_{\nu_A} X_{\nu_A}}{a_{\mu_A} + a_{\nu_A}}.
\end{align}
Recall that $a_{\mu_A}$ were selected as the lowest exponents of an uncontracted $1s$ GTO function located at each atom. Assuming that each element has its own unique AO basis set, the value of $a_{\mu_A}$ depends only on the type of atom (nuclear charge), but not on the location of this particular atom or the number of such atoms in the system. For example, consider a water cluster where there are only two elements in the system, independently of its size, hydrogen and oxygen. The value of $a_{\mu_A}$ is then identical for each hydrogen atom, and a different value $a_{\mu_A}$ is identical for each oxygen atom. More generally, assume that we have $M$ distinct elements in the system. The value of $M$ is completely independent of the system size, as this is necessary to keep the density of the system $\rho$ constant. As a result, we have only $M$ unique values of $a_{\mu_A}$ in the whole molecule, denoted shortly by $a_m$, $m=1,\ldots,M$. To identify them, let us introduce an indicator matrix $I_{\mu_A,m}$ that is equal to $1$ if the function $\mu_A$ belongs to an element with index $m=1,\ldots,M$, and zero otherwise. From this construction, we have for all $\mu_A$ the equality:
\begin{align}
    a_{\mu_A} = \sum_{m=1}^M a_m\,I_{\mu_A,m},
\end{align}
and similarly for any continuous function $f(a_{\mu_A})$ dependent on the exponent. This enables us to rewrite the quantities $X_{\mu_A\nu_A}$ as:
\begin{align}
    X_{\mu_A\nu_A} = \sum_{m,n=1}^M
    \frac{a_m X_{\mu_A} + a_n X_{\nu_A}}{a_m + a_n}\,I_{\mu_A,m}\,I_{\nu_A,n}.
\end{align}
After straightforward rearrangements, we bring this to the form:
\begin{align}
\begin{split}
    X_{\mu_A\nu_A} &= \sum_{m,n=1}^M (X_{\mu_A}\,I_{\mu_A,m})\,
    \frac{a_m}{a_m + a_n}\,I_{\nu_A,n} \\
    &+ \sum_{m,n=1}^M I_{\mu_A,m}\,
    \frac{a_n}{a_m + a_n}\,(X_{\nu_A}\,I_{\nu_A,n}).
\end{split}
\end{align}
To represent this formula in a matrix form that is more convenient later, we introduce matrices $\mathbf{w}$ and $\mathbf{C}$ with elements $w_{mn} = \frac{a_m}{a_m + a_n}$ and $C_{\mu_A,m} = X_{\mu_A}\,I_{\mu_A,m}$. This leads to:
\begin{align}
    \mathbf{X}^A = \mathbf{C}\,\mathbf{w}\,\mathbf{I}^\mathrm{T} + 
    \mathbf{I}\,\mathbf{w}^\mathrm{T}\,\mathbf{C}^\mathrm{T}.
\end{align}
As the dimension of the matrix $\mathbf{w}$ is $M\times M$, we arrive at the conclusion that $\mbox{rank}(\mathbf{X}^A)\leq2M$. The same relationship holds for the remaining two Cartesian components of $\mathbf{R}_{\mu_A\nu_A}$, i.e. $\mbox{rank}(\mathbf{Y}^A)\leq2M$ and $\mbox{rank}(\mathbf{Z}^A)\leq2M$.
Nearly identical manipulations can be performed for the second matrix $\mathbf{M}^A$, showing that $\mbox{rank}(\mathbf{M}^A)\leq5M$. Identical bounds are also valid for the respective matrices $\mathbf{X}^B$ and $\mathbf{M}^B$ related to the sphere $B$. Returning to Eq.~(\ref{eq:grank1}) and using the bounds found above, we obtain:
\begin{align}
\label{eq:grank2}
    \mbox{rank}(\mathbf{G}) \leq 12M^2 + 10M,
\end{align}
demonstrating that the rank of this tensor is bounded from above by a quantity that is independent of the system size. Returning to the example of water clusters with $M=2$, we have $\mbox{rank}(\mathbf{G}) \leq 68$, independently of the number of water molecules present in the spheres $A$ and $B$.

\subsection{\label{sec:tensordminus} Effective rank of the tensor $\mathbf{D}^{\circ -1}$}

In this section, we consider the effective rank of the tensor $\mathbf{D}^{\circ -1}$ and show that this quantity grows very slowly with the system size. The fundamental formula that will be used here is a special case of the Laplace expansion of the interaction potential:
\begin{align}
\label{eq:mult1}
    |\mathbf{r} + \mathbf{R}|^{-1} = \sum_{L=0}^\infty (-1)^L \frac{|\mathbf{r}|^{L}}{|\mathbf{R}|^{L+1}}\,P_L(\cos\theta),
\end{align}
where $\theta$ is the angle between the vectors $\mathbf{r}$ and $\mathbf{R}$, $P_L$ are the Legendre polynomials of degree $L$. This expansion is valid (convergent) for $|\mathbf{r}|\leq |\mathbf{R}|$.

Our plan here is to proceed in three steps. First, we will show how the formula~(\ref{eq:mult1}) can be applied to the elements of the tensor $\mathbf{D}^{\circ -1}$. Next, we will truncate this expansion at some $L=L_{\mathrm{max}}>0$ and study how this approximation translates into the error in the Frobenius norm of the whole tensor as a function of the system size. Finally, we will show that without any further approximations, the truncated expansion~(\ref{eq:mult1}) can be rewritten in the CPD form defined by Eq.~(\ref{eq:cpd-brutal}). This will naturally lead to an upper bound for the effective rank of the tensor $\mathbf{D}^{\circ -1}$.

Note that the tensor $\mathbf{D}^{\circ -1}$ gathers Coulomb interactions between classical unit charges located at points $\mathbf{R}_{\mu_A\nu_A}$ and $\mathbf{R}_{\sigma_B\lambda_B}$. As the parent AO functions with indices $\mu_A,\nu_A$ and $\sigma_B,\lambda_B$ are located within spheres $A$ and $B$, the same is true for the points $\mathbf{R}_{\mu_A\nu_A}$ and $\mathbf{R}_{\sigma_B\lambda_B}$. Recall that the centers of the spheres $A$ and $B$ are located at $\mathbf{O}_A=(0,0,2R/3)$ and $\mathbf{O}_B=(0,0,-2R/3)$. The elements of the tensor $\mathbf{D}^{\circ -1}$ can then be rewritten as:
\begin{align}
\begin{split}
    &|\mathbf{R}_{\mu_A\nu_A}-\mathbf{R}_{\sigma_B\lambda_B}|^{-1}=
    |\mathbf{R}_{\mu_A\nu_A}-\mathbf{O}_A-(\mathbf{R}_{\sigma_B\lambda_B}-\mathbf{O}_B)+\mathbf{R}_0|^{-1} = \\
    &|\mathbf{R}_{\mu_A\nu_A}^{(0)}-\mathbf{R}_{\sigma_B\lambda_B}^{(0)}+\mathbf{R}_0|^{-1},
\end{split}
\end{align}
where $\mathbf{R}_0=(0,0,4R/3)$ is the vector that points from the center of the sphere $A$ to the center of the sphere $B$, while $\mathbf{R}_{\mu_A\nu_A}^{(0)}=\mathbf{R}_{\mu_A\nu_A}-\mathbf{O}_A$ and similarly for the sphere $B$. As the points $\mathbf{R}_{\mu_A\nu_A}^{(0)}$ and $\mathbf{R}_{\sigma_B\lambda_B}^{(0)}$ are located entirely within or on the surfaces of the spheres $A$ and $B$, we have $|\mathbf{R}_{\mu_A\nu_A}^{(0)}|\leq R/3$ and $|\mathbf{R}_{\sigma_B\lambda_B}^{(0)}|\leq R/3$. From the triangle inequality, we then find:
\begin{align}
    |\mathbf{R}_{\mu_A\nu_A}^{(0)}-\mathbf{R}_{\sigma_B\lambda_B}^{(0)}| \leq
    |\mathbf{R}_{\mu_A\nu_A}^{(0)}| + |\mathbf{R}_{\sigma_B\lambda_B}^{(0)}| \leq
    2R/3,
\end{align}
while for the distance between expansion points we have $R_0=|\mathbf{R}_0|=4R/3$. This shows that $|\mathbf{R}_{\mu_A\nu_A}^{(0)}-\mathbf{R}_{\sigma_B\lambda_B}^{(0)}| \leq |\mathbf{R}_0|$ and we can set $\mathbf{r}=\mathbf{R}_{\mu_A\nu_A}^{(0)}-\mathbf{R}_{\sigma_B\lambda_B}^{(0)}$ and $\mathbf{R}=\mathbf{R}_0$ in Eq.~(\ref{eq:mult1}), giving a bipolar form of the Laplace expansion known from the fast multipole method~\cite{rokhlin85,greengard88,greengard94} (FFM) or related techniques:
\begin{align}
\label{eq:mult2}
    |\mathbf{R}_{\mu_A\nu_A}-\mathbf{R}_{\sigma_B\lambda_B}|^{-1} = 
    \sum_{L=0}^\infty (-1)^L \frac{|\mathbf{R}_{\mu_A\nu_A}^{(0)}-\mathbf{R}_{\sigma_B\lambda_B}^{(0)}|^{L}}{R_0^{L+1}}\,P_L(\cos\theta),
\end{align}
where $\theta$ is now the angle between the vectors $\mathbf{R}_{\mu_A\nu_A}^{(0)}-\mathbf{R}_{\sigma_B\lambda_B}^{(0)}$ and $\mathbf{R}_0$.

Let us now assume that the above expansion is truncated at some $L=L_{\mathrm{max}}>0$. This introduces an error which is given by the expression:
\begin{align}
\label{eq:mult3}
\begin{split}
    &\delta |\mathbf{R}_{\mu_A\nu_A}-\mathbf{R}_{\sigma_B\lambda_B}|^{-1} = \\
    &\sum_{L=L_{\mathrm{max}}+1}^\infty (-1)^L \frac{|\mathbf{R}_{\mu_A\nu_A}^{(0)}-\mathbf{R}_{\sigma_B\lambda_B}^{(0)}|^{L}}{R_0^{L+1}}\,P_L(\cos\theta).
\end{split}
\end{align}
We want to derive an upper bound for the value of this error as a function of $L_{\mathrm{max}}$. To this end, we first use the fact that the sum is necessarily smaller than the sum of absolute values of each term, giving:
\begin{align}
\label{eq:mult4}
\begin{split}
    \delta |\mathbf{R}_{\mu_A\nu_A}-\mathbf{R}_{\sigma_B\lambda_B}|^{-1} \leq 
    \sum_{L=L_{\mathrm{max}}+1}^\infty \frac{|\mathbf{R}_{\mu_A\nu_A}^{(0)}-\mathbf{R}_{\sigma_B\lambda_B}^{(0)}|^{L}}{R_0^{L+1}}\,|P_L(\cos\theta)|.
\end{split}
\end{align}
Next, we use the upper bound for the Legendre polynomials $|P_L(x)|\leq 1$ which is valid everywhere within the domain $x\in[-1,+1]$, together with the inequality $|\mathbf{R}_{\mu_A\nu_A}^{(0)}-\mathbf{R}_{\sigma_B\lambda_B}^{(0)}|\leq 2R/3$ derived above. Additionally, we insert $R_0=4R/3$ and after some rearrangements we arrive at a simplified upper bound for the sum:
\begin{align}
\label{eq:mult5}
    \delta |\mathbf{R}_{\mu_A\nu_A}-\mathbf{R}_{\sigma_B\lambda_B}|^{-1} \leq 
    \frac{3}{4R}\sum_{L=L_{\mathrm{max}}+1}^\infty \frac{1}{2^L}.
\end{align}
The remaining sum is an elementary geometric series and can be evaluated exactly:
\begin{align}
\label{eq:mult6}
    \delta |\mathbf{R}_{\mu_A\nu_A}-\mathbf{R}_{\sigma_B\lambda_B}|^{-1} \leq 
    \frac{3}{4R}\,2^{-L_{\mathrm{max}}}.
\end{align}
This shows that the error of the truncated expansion is bounded from above by a quantity that is independent of the indices of the tensor. However, we are not interested in the elementwise error, but rather in the Frobenius norm of the error tensor. To derive it, we simply apply the upper bound found above to each entry of the error tensor, sum over all its elements, and take the square root to find:
\begin{align}
\label{eq:mult7}
    ||\delta |\mathbf{R}_{\mu_A\nu_A}-\mathbf{R}_{\sigma_B\lambda_B}|^{-1} || \leq 
    \frac{3}{4R}\,2^{-L_{\mathrm{max}}} N_{\mathrm{sub}}^A\,N_{\mathrm{sub}}^B,
\end{align}
taking into account the dimensions of the tensor along each mode. A more convenient upper bound is established by additionally using the facts that $N_{\mathrm{sub}}^A\leq N_{\mathrm{AO}}$, $N_{\mathrm{sub}}^B\leq N_{\mathrm{AO}}$, and $R=\rho N_{\mathrm{AO}}^{1/3}$, leading to:
\begin{align}
\label{eq:mult8}
    ||\delta |\mathbf{R}_{\mu_A\nu_A}-\mathbf{R}_{\sigma_B\lambda_B}|^{-1} || \leq 
    \frac{3 N_{\mathrm{AO}}^{5/3}}{4\rho}\,2^{-L_{\mathrm{max}}}.
\end{align}
According to the discussion in Sec.~\ref{sec:dhada}, we have to find a value of $L_{\mathrm{max}}$ that is sufficient to reach an error no greater than $\frac{\epsilon}{(2R)^2} = \frac{\epsilon}{4\rho^2 N_{\mathrm{AO}}^{2/3}}$ for $0<\epsilon<1$ in the Frobenius norm, i.e. $||\delta |\mathbf{R}_{\mu_A\nu_A}-\mathbf{R}_{\sigma_B\lambda_B}|^{-1} || \leq \frac{\epsilon}{4\rho^2 N_{\mathrm{AO}}^{2/3}}$. After straightforward rearrangements, we obtain:
\begin{align}
    L_{\mathrm{max}} \geq \log_2\left( \frac{3\rho N_{\mathrm{AO}}^{7/3}}{\epsilon}\right),
\end{align}
and this reveals that to satisfy the error requirement, it is sufficient to set:
\begin{align}
\label{eq:lmax}
    L_{\mathrm{max}} = c_\epsilon \log_2 N_{\mathrm{AO}},
\end{align}
where $c_\epsilon$ is a constant that is independent of the system size, given by the formula:
\begin{align}
    c_\epsilon = \frac{7}{3} + \max\left( 0, \log_2 \frac{3\rho}{\epsilon} \right).
\end{align}
In other words, to keep the error~(\ref{eq:mult3}) of the truncated bipolar expansion below a threshold $\frac{\epsilon}{(2R)^2}$ (in the Frobenius norm), the parameter $L_{\mathrm{max}}$ has to increase at worst logarithmically with the system size.

Once the expansion in Eq.~(\ref{eq:mult2}) is truncated at $L_{\mathrm{max}}$, it can be transformed into the CPD approximation defined through Eq.~(\ref{eq:cpd-brutal}) without introducing any additional approximations. However, this transformation requires one to take several steps. First, by virtue of the spherical harmonic addition theorem, the factor $P_L(\cos\theta)$ is expanded as:
\begin{align}
    P_L(\cos\theta) = \frac{4\pi}{2L+1}\sum_{M=-L}^L Y_{LM}(\mathbf{R}_{\mu_A\nu_A}^{(0)}-\mathbf{R}_{\sigma_B\lambda_B}^{(0)})\,Y_{LM}^*(\mathbf{R}_0),
\end{align}
where $Y_{lm}$ are the spherical harmonics. As the vector $\mathbf{R}_0$ is aligned along the $z$ axis, only the term with $M=0$ brings a non-zero contribution to the sum. This allows us to simplify the original expansion to the form:
\begin{align}
\label{eq:mult8b}
\begin{split}
    |\mathbf{R}_{\mu_A\nu_A}-\mathbf{R}_{\sigma_B\lambda_B}|^{-1} &= 
    \sum_{L=0}^{L_{\mathrm{max}}} (-1)^L\frac{4\pi}{2L+1} \frac{|\mathbf{R}_{\mu_A\nu_A}^{(0)}-\mathbf{R}_{\sigma_B\lambda_B}^{(0)}|^{L}}{R_0^{L+1}} \\
    &\times Y_{L0}(\mathbf{R}_{\mu_A\nu_A}^{(0)}-\mathbf{R}_{\sigma_B\lambda_B}^{(0)})\,
    Y_{L0}(\mathbf{R}_0),
\end{split}
\end{align}
where we have also dropped the complex conjugation operation, because the spherical harmonics with $M=0$ are purely real. To bring this formula into a more compact form, convenient for subsequent manipulations, we write it in terms of the regular and irregular solid harmonics defined as:
\begin{align}
    &\mathcal{R}_{lm}(\mathbf{r}) = \sqrt{\frac{4\pi}{2l+1}} r^l Y_{lm}(\mathbf{r}), \\
    &\mathcal{I}_{lm}(\mathbf{r}) = \sqrt{\frac{4\pi}{2l+1}} \frac{1}{r^{l+1}} Y_{lm}(\mathbf{r}),
\end{align}
which gives:
\begin{align}
\label{eq:mult9}
\begin{split}
    |\mathbf{R}_{\mu_A\nu_A}-\mathbf{R}_{\sigma_B\lambda_B}|^{-1} &= 
    \sum_{L=0}^{L_{\mathrm{max}}} (-1)^L\mathcal{I}_{L0}(\mathbf{R}_0)\,\mathcal{R}_{L0}(\mathbf{R}_{\mu_A\nu_A}^{(0)}-\mathbf{R}_{\sigma_B\lambda_B}^{(0)}).
\end{split}
\end{align}
The pairs of indices $\mu_A\nu_A$ and $\sigma_B\lambda_B$ can be separated using the formula for the translation of the regular solid harmonics. In the general case, it reads:
\begin{align}
\label{eq:rlmtrans}
    \mathcal{R}_{LM}(\mathbf{r}_1 - \mathbf{r}_2) = \sum_{l=0}^L \sum_{m=-l}^l C_{lm}^{LM}\,
    \mathcal{R}_{lm}(\mathbf{r}_1)\,\mathcal{R}_{L-l,M-m}(\mathbf{r}_2),
\end{align}
where the coefficients $C_{lm}^{LM}$ are not written explicitly to avoid notational clutter, but can be found elsewhere~\cite{sack64,steinborn73}. Note that this expansion is finite and hence introduces no approximation. Applying it to the term $\mathcal{R}_{L0}(\mathbf{R}_{\mu_A\nu_A}^{(0)}-\mathbf{R}_{\sigma_B\lambda_B}^{(0)})$ in Eq.~(\ref{eq:mult9}) we obtain:
\begin{align}
\label{eq:mult10}
\begin{split}
    |\mathbf{R}_{\mu_A\nu_A}-\mathbf{R}_{\sigma_B\lambda_B}|^{-1} &= 
    \sum_{L=0}^{L_{\mathrm{max}}} (-1)^L\mathcal{I}_{L0}(\mathbf{R}_0)\,
    \sum_{l=0}^L \sum_{m=-l}^l C_{lm}^{L0} \\
    &\times
    \mathcal{R}_{lm}(\mathbf{R}_{\mu_A\nu_A}^{(0)})\,\mathcal{R}_{L-l,-m}(\mathbf{R}_{\sigma_B\lambda_B}^{(0)}),
\end{split}
\end{align}
However, this expansion cannot yet be matched with the CPD format, cf. Eq.~(\ref{eq:cpd-brutal}), because pairs of indices $\mu_A\nu_A$ and $\sigma_B\lambda_B$ in Eq.~(\ref{eq:mult1}) are still ``pinned'' together in the regular solid harmonics rather than appearing fully independently as in Eq.~(\ref{eq:cpd-brutal}). Fortunately, separation of these indices can also be achieved without introducing any extra error. To this end, we first recall that $\mathbf{R}_{\mu_A\nu_A} = \frac{a_{\mu_A} \mathbf{R}_{\mu_A} + a_{\nu_A} \mathbf{R}_{\nu_A}}{a_{\mu_A} + a_{\nu_A}}$ which allows us to write:
\begin{align}
\begin{split}
    &\mathbf{R}_{\mu_A\nu_A}^{(0)} = \mathbf{R}_{\mu_A\nu_A} - \mathbf{O}_A = 
    \frac{a_{\mu_A} \mathbf{R}_{\mu_A} + a_{\nu_A} \mathbf{R}_{\nu_A}}{a_{\mu_A} + a_{\nu_A}} - \mathbf{O}_A = \\
    &\frac{a_{\mu_A}}{a_{\mu_A} + a_{\nu_A}}(\mathbf{R}_{\mu_A} - \mathbf{O}_A) +
     \frac{a_{\nu_A}}{a_{\mu_A} + a_{\nu_A}}(\mathbf{R}_{\nu_A} - \mathbf{O}_A) = \\
    &\frac{a_{\mu_A} \mathbf{R}_{{\mu_A}}^{(0)} + a_{\nu_A} \mathbf{R}_{{\nu_A}}^{(0)}}{a_{\mu_A} + a_{\nu_A}}.
\end{split}
\end{align}
As the regular solid harmonic is a homogeneous polynomial in the Cartesian components of the vector argument, we have:
\begin{align}
\begin{split}
    &\mathcal{R}_{lm}(\mathbf{R}_{\mu_A\nu_A}^{(0)}) = 
    \mathcal{R}_{lm}\left(\frac{a_{\mu_A} \mathbf{R}_{{\mu_A}}^{(0)} + a_{\nu_A} \mathbf{R}_{{\nu_A}}^{(0)}}{a_{\mu_A} + a_{\nu_A}}\right) = \\
    &\frac{1}{(a_{\mu_A} + a_{\nu_A})^l}\,\mathcal{R}_{lm}\left(a_{\mu_A} \mathbf{R}_{{\mu_A}}^{(0)} + a_{\nu_A} \mathbf{R}_{{\nu_A}}^{(0)}\right).
\end{split}
\end{align}
Using Eq.~(\ref{eq:rlmtrans}), the second term is now straightforward to separate into a sum of products of quantities that depend on the indices $\mu_A$ and $\nu_A$, namely:
\begin{align}
\begin{split}
    &\mathcal{R}_{lm}\left(a_{\mu_A} \mathbf{R}_{{\mu_A}}^{(0)} + a_{\nu_A} \mathbf{R}_{{\nu_A}}^{(0)}\right) =
    \sum_{j_A=0}^l \sum_{k_A=-j_A}^{j_A} C_{j_Ak_A}^{lm} \\
    &\times\mathcal{R}_{j_Ak_A}\left(a_{\mu_A} \mathbf{R}_{{\mu_A}}^{(0)}\right)\,
    \mathcal{R}_{l-j_A,m-k_A}\left(-a_{\nu_A} \mathbf{R}_{{\nu_A}}^{(0)}\right).
\end{split}
\end{align}
To achieve an analogous separation for the first term, we use the indicator function trick introduced in Sec.~\ref{sec:tensord2}. Preserving the same notation, we have:
\begin{align}
    \frac{1}{(a_{\mu_A} + a_{\nu_A})^l} = \sum_{mn} I_{\mu_A,m}\,\frac{1}{(a_m + a_n)^l}\,I_{\nu_A,n},
\end{align}
where the summation runs over unique atom types in the molecule, and hence its length is independent of the system size. By performing analogous manipulations for the term involving the pair of indices $\sigma_B\lambda_B$ and inserting all results back into Eq.~(\ref{eq:mult10}) we finally arrive at after some rearrangements:
\begin{widetext}
\begin{align}
\label{eq:mult11}
\begin{split}
    &|\mathbf{R}_{\mu_A\nu_A}-\mathbf{R}_{\sigma_B\lambda_B}|^{-1} = 
    \sum_{m_1n_1} \sum_{m_2n_2} 
    \sum_{L=0}^{L_{\mathrm{max}}} 
    \frac{(-1)^L\mathcal{I}_{L0}(\mathbf{R}_0)}{(a_{m_2} + a_{n_2})^{L}}
    \sum_{l=0}^L \sum_{m=-l}^l C_{lm}^{L0}\,
    \left(\frac{a_{m_2} + a_{n_2}}{a_{m_1} + a_{n_1}}\right)^l \\
    &\sum_{j_A=0}^l \sum_{k_A=-j_A}^{j_A} C_{j_Ak_A}^{lm}\,
    \left[ I_{\mu_A,m_1}\,\mathcal{R}_{j_Ak_A}\left(a_{\mu_A} \mathbf{R}_{{\mu_A}}^{(0)}\right)\right]
    \left[ I_{\nu_A,n_1}\,\mathcal{R}_{l-j_A,m-k_A}\left(-a_{\nu_A} \mathbf{R}_{{\nu_A}}^{(0)}\right)\right] \\
    &\sum_{j_B=0}^{L-l} \sum_{k_B=-j_B}^{j_B} C_{j_Bk_B}^{L-l,-m}\,
    \left[ I_{\sigma_B,m_2}\,\mathcal{R}_{j_Bk_B}\left(a_{\sigma_B} \mathbf{R}_{{\sigma_B}}^{(0)}\right)\right]
    \left[ I_{\lambda_B,n_2}\,\mathcal{R}_{L-l-j_B,-m-k_B}\left(- a_{\lambda_B} \mathbf{R}_{{\lambda_B}}^{(0)}\right)\right].
\end{split}
\end{align}
\end{widetext}
Note that in this equation, all four indices of the original tensor ($\mu_A\nu_A\sigma_B\lambda_B$) are completely separated, and therefore this approximation can be matched with the CPD expansion given in Eq.~(\ref{eq:cpd-brutal}). As a consequence, we are guaranteed that the effective rank of the tensor $\mathbf{D}^{\circ -1}$ is not larger than the number of terms in Eq.~(\ref{eq:mult11}). Therefore, it remains to count the number of terms in this expression. The length of the first four summations (over $m_1,n_1,m_2,n_2$) is independent of the system size. They collectively introduce only a constant prefactor dependent solely on the number of unique atom types in the system, but are not affected by the truncation parameter $L_{\mathrm{max}}$. Therefore, the number of terms is:
\begin{align}
    \mbox{const}\cdot\sum_{L=0}^{L_{\mathrm{max}}} \sum_{l=0}^L \sum_{m=-l}^l
    \sum_{j_A=0}^l \sum_{k_A=-j_A}^{j_A} \sum_{j_B=0}^{L-l} \sum_{k_B=-j_B}^{j_B} 1.
\end{align}
This nested sum can be evaluated exactly, but what matters from our perspective is that the highest power of $L_{\mathrm{max}}$ that appears in the final expression is $L_{\mathrm{max}}^7$. By recalling that $L_{\mathrm{max}} = c_\epsilon \log_2 N_{\mathrm{AO}}$ we conclude that there exists a constant $c_\epsilon^{(7)}$ independent of the system size such that:
\begin{align}
    \mbox{rank}_{\epsilon/(2R)^2}(\mathbf{D}^{\circ -1}) \leq c_\epsilon^{(7)}\,\log_2^7 N_{\mathrm{AO}}
\end{align}
which is the main result of the present section.

Before moving on, we have to clarify one issue related to the intermediate result expressed in Eq.~(\ref{eq:lmax}), i.e. in order to keep the error of the truncated multipole expansion below a predefined threshold $\epsilon$, the length of the expansion $L_{\mathrm{max}}$ has to increase logarithmically. This might be confusing for readers familiar with the fast multipole method~\cite{rokhlin85,greengard88,greengard94} or its continuous analogues~\cite{white94,white96}, where the length of the multipole expansion used in practice is completely independent of the system size. For example, interaction between two groups of well-separated classical charges is calculated using a multipole expansion with the number of terms dependent only on the desired accuracy level ($\epsilon$), but not on the number of charges (as would be in our case due to the presence of the logarithmic term in $L_{\mathrm{max}}$). The reason for this apparent discrepancy lies entirely within the metric used to quantify the error. In our case, the error is assessed by calculating the Frobenius norm of the error tensor, $\delta |\mathbf{R}_{\mu_A\nu_A}-\mathbf{R}_{\sigma_B\lambda_B}|^{-1}$, see Eqs.~(\ref{eq:mult7})~and~(\ref{eq:mult8}). In FFM and related techniques, one ensures that the maximum absolute error among each individual pair of interacting charges is below a predefined threshold. In our language, this would be equivalent to using the max norm, i.e. $||\mathbf{T}||=\mbox{max}(|\mathbf{T}|)$ rather than the Frobenius norm in the determination of the error. The use of the max norm is a completely legitimate way to measure the error, but it is impractical in the context of CPD, because there are no algorithms with efficiency and simplicity comparable to ALS to optimize the CPD expansion of ERI in norms other than the Frobenius norm. In other fields, there are known applications~\cite{gandy11,shi21,lu21} of other norms, in particular the nuclear tensor norm~\cite{lim10,lim13,friedland18} or $L_1$ norm, to optimize the CPD for some tensors. However, to our knowledge, in the case of ERI such an approach appears to be extremely complicated and computationally costly, and has never been attempted thus far.

\subsection{\label{sec:tensord} Effective rank of the tensor $\mathbf{D}$}

We can now return to Eq.~(\ref{eq:rankd1}) to finalize the discussion of the effective rank of the tensor $\mathbf{D}$. From Sec.~\ref{sec:tensord2} we know that the exact rank of the tensor $\mathbf{G}$ is independent of the system size, while from Sec.~\ref{sec:tensordminus} we found an upper bound for the effective rank of the tensor $\mathbf{D}^{\circ -1}$. Putting these findings together in Eq.~(\ref{eq:rankd1}) we conclude that there exists a constant $d_\epsilon^{(7)}$ such that:
\begin{align}
    \mbox{rank}_{\epsilon}(\mathbf{D}) \leq d_\epsilon^{(7)}\,\log_2^7 N_{\mathrm{AO}}.
\end{align}
At this point, we have all necessary information to quantify the rank of the ERI tensor.

\subsection{\label{sec:summary} Final lower bound for the rank of ERI tensor}

Now we gather all results from the previous sections and determine the final conclusions for the CPD rank of the ERI tensor as a function of the system size. First, we recall the inequality derived in Sec.~\ref{sec:subprop}:
\begin{align}
\label{eq:con1}
    \mbox{rank}_{\epsilon-\delta}(\mathbf{T}) \geq 
    \frac{\mathrm{rank}_\eta(\mathbf{N})}{\mathrm{rank}_\epsilon(\mathbf{D})}.
\end{align}
From the results of Sec.~\ref{sec:tensord} we know that there exists a constant $d_\epsilon^{(7)}$ such that $\mathrm{rank}_\epsilon(\mathbf{D})\leq d_\epsilon^{(7)}\,\log_2^7 N_{\mathrm{AO}}$ for any $0<\epsilon<1$ independently of the system size. Simultaneously, from Sec.~\ref{sec:tensorn} we know that $\mbox{rank}_\eta(\mathbf{N}) > c_\eta^{(2)} N_{\mathrm{AO}}^2$ for any $0<\eta<1$. This brings us close to establishing a lower bound for the rank of $\mathbf{T}$, but since the value of $\eta$ is different from $\epsilon$ we have to analyze the conditions under which $\eta<1$ such that the results from Sec.~\ref{sec:tensorn} hold. Recall that $\eta=\big[\mbox{max}(|\mathbf{N}\odot \mathbf{D}^{\circ -1}|)+\mbox{max}(|\mathbf{D}|)\big]\epsilon+\epsilon^2$ and the requirement that:
\begin{align}
    \eta=\big[\mbox{max}(|\mathbf{N}\odot \mathbf{D}^{\circ -1}|)+\mbox{max}(|\mathbf{D}|)\big]\epsilon+\epsilon^2 < 1,
\end{align}
can be replaced by a stronger one:
\begin{align}
    \epsilon < \frac{1}{\mbox{max}(|\mathbf{N}\odot \mathbf{D}^{\circ -1}|)+\mbox{max}(|\mathbf{D}|)+1}.
\end{align}
If the latter inequality holds, then the requirement $\eta<1$ is automatically satisfied. Our next goal is to relate the denominator in the above expression to the spatial characteristics of the model system. Because the tensor $\mathbf{D}$ gathers distances between points within spheres $A$ and $B$ we have $\mbox{max}(|\mathbf{D}|)\leq 2R$ and $\mbox{min}(|\mathbf{D}|)\geq 2R/3$. The largest element of the tensor $\mathbf{N}$ is also easily bounded by noting that it consists of products of overlap integrals that cannot exceed $1$ in value due to the normalization condition of the AO basis functions. Therefore, we have $\mbox{max}(|\mathbf{N}|)\leq 1$ and $\mbox{max}(|\mathbf{N}\odot \mathbf{D}^{\circ -1}|)\leq \frac{3}{2R}$. This leads to an explicit bound for the $\epsilon$ parameter:
\begin{align}
    \epsilon < \frac{1}{\frac{3}{2R}+2R+1}.
\end{align}
Therefore, for any value of $\epsilon$ such that:
\begin{align}
    \frac{1}{\frac{3}{2R}+2R+1} > \epsilon > \delta,
\end{align}
we have $\eta<1$ and hence the inequality in Eq.~(\ref{eq:con1}) simplifies to the form:
\begin{align}
\label{eq:con2}
    \mbox{rank}_{\epsilon-\delta}(\mathbf{T}) > 
    \frac{c_\eta^{(2)}\,N_{\mathrm{AO}}^2}{d_\epsilon^{(7)}\,\log_2^7 N_{\mathrm{AO}}}.
\end{align}
This is equivalent to the main Theorem~\ref{th:main} of this work after some rearrangements.

\section{\label{sec:discussion} Discussion}

\subsection{\label{sec:example} Practical example}

To illustrate the conditions under which Theorem~\ref{th:main} holds in a practical situation, let us once again consider the example of water clusters. As an example, consider the cc-pVDZ AO basis set~\cite{dunning89} in which $a_{\mathrm{min}}=0.122$ (uncontracted 1s GTO of the hydrogen atom). The parameter $\rho$ is on the order of unity, so we set $\rho=1$ to simplify the discussion.  From previous applications of the CPD approximation available in the literature~\cite{benedikt11,pierce21}, a reasonable value of the parameter $\epsilon$ that is sufficient in most application is $\epsilon=10^{-3}$. Consider the constraints imposed on the value of $\epsilon$ by the statement of Theorem~\ref{th:main}. First, the lower bound on $\epsilon$ is given by the parameter $\delta$. Even for a modest size of the system $R=20$, $\delta\approx 2\cdot10^{-4}$, and for $R=25$ we have $\delta\approx 3\cdot10^{-9}$, while for $R=30$ the value of $\delta$ is at the level of numerical noise (in double precision arithmetic). The radius $R=20$ corresponds to roughly 150 water molecules, which is perfectly within the realm of modern quantum chemical simulations, even at the correlated level. Consider now the upper bound, $\epsilon < \frac{1}{\frac{3}{2R}+2R+1}$. Assuming that we wish to apply $\epsilon=10^{-3}$, the radius of the system must be less than roughly $R=500$. This corresponds to a comically large water cluster comprising roughly 2.5 million molecules. This shows that the applicability range of Theorem~\ref{th:main} is indeed broad and covers molecules that are frequently encountered in practice.

\subsection{\label{sec:origin} Origin of the rank explosion}

If we wanted to briefly summarize what the origins of the rank explosion of the CPD decomposition are, i.e. why the effective rank cannot grow linearly with the system size, the reasons can be traced back to the fundamental inability of the global CPD approximation to reproduce the multipole expansion of the long-range ERI while keeping the rank linear (combined with the fact that the Coulomb integrals vanish very slowly as a function of the distance between charge distributions). In fact, the proof of Theorem~\ref{th:main} presented in Sec.~\ref{sec:lower} is based on selecting a certain subset of ERI which can be approximated sufficiently accurately by the first term of the multipole expansion, i.e. monopole-monopole term; see the definition of $\mathbf{T}_{\mathrm{sub}}^{\mathrm{(mon)}}$ in Eq.~(\ref{eq:tmon}). This is somewhat counterintuitive because we know from FFM or related techniques that such long-range contributions can be evaluated with $\mathcal{O}(N)$ or $\mathcal{O}(N\log N)$ computational cost. However, the difference lies in the fact that the FFM operates at a local level, while the CPD attempts to provide a global approximation for the entire ERI tensor. More precisely, in the FFM approach, the system is subdivided into boxes. The pairs of boxes which are sufficiently well-separated are treated at the highest possible level, i.e. without subdividing them further. However, each pair of such well-separated boxes is treated using their own multipole expansion carried out at a local level, i.e. the expansion points of the multipole expansion coincide with the centers of the boxes for a given pair. Therefore, in the FFM approach, no single global form of the approximation is utilized. By contrast, the CPD format attempts to find a single approximation that holds equally well for all classes of integrals. As a side note, we mention that it is possible to envisage a combination of FFM and CPD, where the system is first divided into boxes analogously as in FFM and then a local CPD approximation is established for each pair of boxes. This will likely avoid the rank explosion, but the resulting approximation will no longer be global. This would distinguish it from well-known ERI formats such as density-fitting, Cholesky decomposition, or tensor hypercontraction, and likely make it less appealing in practice, especially at the correlated level.

\subsection{\label{sec:implication} Implications for use of CPD in quantum chemistry}

The fact that the effective rank of the CPD approximation for ERI cannot increase linearly with the system size is, of course, a significant setback. Nevertheless, we argue that the CPD approximation can still be extremely useful in many applications in quantum chemistry. Indeed, in some situations, rank explosion can be avoided by preparing a CPD suited for a particular application. As an example, consider the construction of the exchange contribution to the Fock matrix which is the major bottleneck in Hartree-Fock and hybrid DFT calculations for large systems~\cite{challacombe97,ochsenfeld98}. The exchange part is defined as:
\begin{align}
    F_{\mu\nu}^{\mathrm{ex}} = -\sum_{\sigma\lambda} D_{\sigma\lambda} (\mu\sigma|\nu\lambda),
\end{align}
where $D_{\sigma\lambda}$ is the density matrix in the AO basis. It is well-known that for insulators, the density matrix vanishes exponentially with the distance between the centers at which the AO orbitals with indices $\sigma\lambda$ are located. Therefore, all ERI $(\mu\sigma|\nu\lambda)$ in the above formula in which the orbitals $\sigma\lambda$ are sufficiently far away will bring no appreciable contribution to the Fock matrix, independent of their actual magnitude. The same is true when the exact integrals $(\mu\sigma|\nu\lambda)$ in the above formula are replaced by the CPD approximation, Eq.~(\ref{eq:cpd-brutal}). Therefore, when the CPD approximation is optimized, the subset of integrals $(\mu\sigma|\nu\lambda)$ where the AO orbitals $\sigma\lambda$ are separated by a large distance can be completely neglected or treated crudely, because this will have no impact on the accuracy of the Fock matrix construction. This will naturally lead to a reduction of the rank of the CPD expansion, possibly even to linear with the system size, because $\mathcal{O}(N_{\mathrm{AO}}^2)$ pairs of distant orbitals are now excluded. Only pairs $\sigma\lambda$ corresponding to orbitals that are in close spatial proximity to each other must be retained, and the number of such ``strong'' pairs increases as $\mathcal{O}(N_{\mathrm{AO}})$.

Of course, a CPD approximation obtained in this way is not universal. Although it is likely to preserve a reasonable accuracy level in the computation of the exchange contribution to the Fock matrix, it will lead to huge errors when applied to the construction of the Coulomb part. Therefore, a specific CPD approximation can be designed with a specific goal in mind rather than as a universal tool that can be applied uniformly across all quantum chemistry methods. This is a considerable departure from methods such as Cholesky decomposition or tensor hypercontraction which are not biased towards (or at least can be unbiased) any particular task in quantum chemistry.

Finally, note that our treatment of the CPD approximation was based on ERI expressed in the AO basis. However, in many calculations in quantum chemistry, especially at the correlated level, one employs ERI in the molecular orbitals (MO) basis. Let us denote the occupied MO by the indices $i,j,k,\ldots$ and the virtual MO by $a,b,c,\ldots$ As discussed above, the inability of the CPD format with a linear rank to represent ERI sufficiently accurately can be traced back to the subtensor $\mathbf{T}_{\mathrm{sub}}^{\mathrm{(mon)}}$, Eq.~(\ref{eq:tmon}), which gathers a specific subset of long-range integrals. This class of integrals vanishes too slowly to become negligible for systems of size reachable in practice and, due to their structure, superlinear CPD rank is needed to represent them accurately. However, consider now the $(ia|jb)$ subclass of ERI in the MO basis. For this class of integrals, the overlap integrals $S_{ia}$ vanish due to the orthonormality of the MO. Therefore, at long-range these integrals are not dominated by the monopole-monopole term of the multipole expansion but rather by the analogous dipole-dipole term. The latter term vanishes much faster ($1/R^3$) compared to the former ($1/R$) as a function of the distance between the charge distributions $ia$ and $jb$. Therefore, the upper bound for $\epsilon$ present in Theorem~\ref{th:main} is no longer valid for $(ia|jb)$, and in this case a linear CPD rank could be achieved for systems of reasonable size. This is potentially important for some quantum-chemistry methods, such as second-order perturbation theory (MP2) or random-phase approximation (RPA), where only certain subclasses of ERI are needed to evaluate the correlation energy~\cite{scuseria08,eshuis10}. However, if we go to higher-order methods, such as coupled cluster~\cite{crawford07,bartlett07}, all possible classes of MO integrals are generally required. While the $(ia|jb)$ class vanishes as $1/R^3$, the $(ia|bc)$ and $(ia|jk)$ classes vanish as $1/R^2$ (monopole-dipole), while the remaining: $(ij|kl)$, $(ij|ab)$, and $(ab|cd)$ vanish as $1/R$ -- the same rate as the integrals in the AO basis. Therefore, when all classes of integrals are required, there is no immediate benefit in switching from AO to MO basis representation. Of course, it is worth pointing out that AO-MO basis transformation is not free of computational cost, even if only a subset of MO integrals is required in a particular application.

\subsection{\label{sec:other} Comparison with other low-rank approximations}

To put the findings of this work in a proper context, it is useful to compare CPD with other low-rank ERI decompositions (see Refs.~\onlinecite{hoy13,hoy15} for an overview of decompositions used in the literature). Of particular relevance here is the tensor hypercontraction~\cite{hohenstein12,parrish12} (THC) which takes the form:
\begin{align}
\label{eq:thc}
    (\mu\nu|\sigma\lambda) = \sum_{RS} X_{\mu R}\,X_{\nu R}\,Z_{RS}\,X_{\sigma S}\,X_{\lambda S}.
\end{align}
Most often, the factors $X_{\mu R}$ are equal to weighted values of AO basis functions on a three-dimensional spatial grid, while the matrix $Z_{RS}$ is obtained by least-squares fitting to minimize the error in the Frobenius norm, see Ref.~\onlinecite{matthews20} for details. As the size of the numerical grid is tied to the number of atoms in the system, the length of each summation over $R$ and $S$ in Eq.~(\ref{eq:thc}) increases strictly linearly with the system size. It is important to realize that the THC factorization can be rewritten in the CPD form, Eq.~(\ref{eq:cpd-brutal}), without loss of generality and without introducing any extra error. To this end, each pair of the $RS$ indices is taken as one of the factors in Eq.~(\ref{eq:cpd-brutal}) and assigned a collective index $r$, so CPD is just a ``vectorized'' form of THC in this particular case. Moreover, since Eq.~(\ref{eq:thc}) preserves all symmetries of the two-electron integrals, the same would be true for the resulting CPD. This means that by proper rearrangements, THC can be put in the form~(\ref{eq:cpd-symm}) which we postulated as a special form of CPD adapted to the natural symmetries of ERI. This connection between CPD and THC also establishes an upper bound for the rank of CPD of ERI. Once THC is ``vectorized'', each pair $RS$ constitutes a single factor in Eq.~(\ref{eq:cpd-brutal}), so the rank of CPD becomes strictly quadratic in the system size. Of course, this way of generating CPD directly from THC is not practical, as it offers no computational benefits over THC. Nonetheless, if a deterministic CPD algorithm for ERI is ever found, it should deliver the rank bounded from below by $\propto N_{\mathrm{AO}}^2/\log_2^7 N_{\mathrm{AO}}$, see Theorem~\ref{th:main}, and from above by $\propto N_{\mathrm{AO}}^2$. It is worth mentioning that when CPD is found numerically such as with ALS or more advanced algorithm, the resulting decomposition does not need to follow this upper bound as ALS is not guaranteed to find the best possible approximation of a given rank due to initialization issues and possibility of encountering local minima.

Another decomposition that is closely related to CPD is the so-called double factorization~\cite{peng17,motta19} that has the general form:
\begin{align}
    (\mu\nu|\sigma\lambda) = \sum_{tRS} X_{\mu R}^t\,X_{\nu R}^t\,Z_{RS}^t\,X_{\sigma S}^t\,X_{\lambda S}^t.
\end{align}
and has gained popularity in recent years due to its applications in quantum computing algorithms, see Refs.~\onlinecite{motta21,cohn21,hohenstein23} and references therein. Using a similar ``vectorization'' strategy, this decomposition can also be written in the CPD format~(\ref{eq:cpd-brutal}). However, since the factors $X_{\mu R}^t$ in this expansion are typically forced to obey certain orthonormality constraints, the rank in the double factorization is typically higher than of THC. To the best of our knowledge, THC provides the tightest upper bound for the CPD rank out of all ERI factorizations proposed in literature.

\section{\label{sec:numerical} Numerical examples}

\begin{figure}
    \centering
    \hspace{-0.5cm}\includegraphics[width=0.7\linewidth]{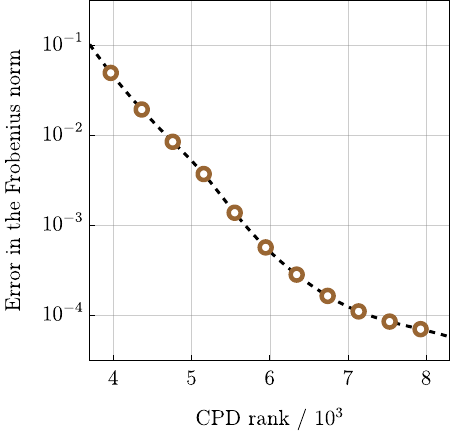}\vspace{0.5cm} \\
    \hspace{-0.5cm}\includegraphics[width=0.7\linewidth]{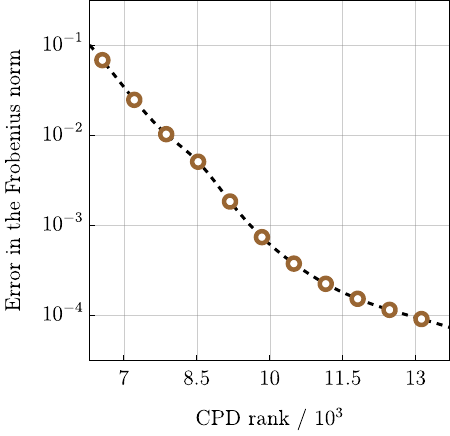}\vspace{0.5cm} \\
    \hspace{-0.5cm}\includegraphics[width=0.7\linewidth]{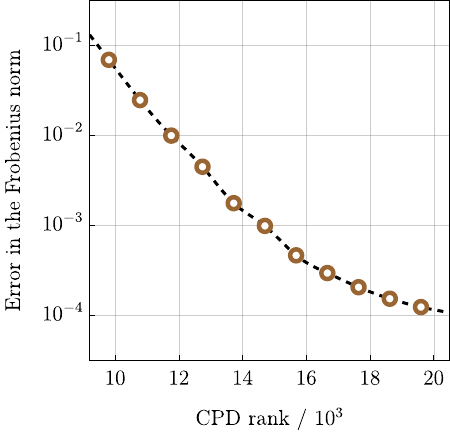}
    \caption{Error of the CPD expansion (logarithmic scale on the $y$ axis) as a function of the CPD rank for the model system described in Sec.~\ref{sec:numerical}. The top, middle, and bottom panels correspond to $n=21,27,33$, respectively, water molecules within the spheres~$A$~and~$B$. The black dashed lines represent the third-order $B$-splines interpolation of the data points.}
    \label{fig:errors}
\end{figure}

In this section, we numerically illustrate the growth of the CPD rank for a specific ERI subtensor for an actual physical system. The goal is to determine an empirical function that captures the scaling behavior of the CPD rank vs. the system size, and contrast it with the lower bound found by mathematical means. First, we construct a minimalist analog of the $\mathbf{T}_{\mathrm{sub}}$ subtensor defined in Eq.~(\ref{eq:tsub}). To this end, we take a set of water clusters (H$_2$O)$_n$ with $n=3,6,\ldots, 36$. Next, we find the minimal radii $R_n/3$ of a sphere that completely encloses a cluster with $n$ water molecules. The center of this sphere coincides with the geometric center of the cluster and the system is translated so that this center lies on the $z$ axis at position $(0,0,2R_n/3)$ for each $n$. This setup is identical to the construction described in Sec.~\ref{sec:system} and hence we refer to this cluster by the letter $A$ and similarly for the sphere $A$ with center at $(0,0,2R_n/3)$. Next, we generate the cluster $B$ by assuming that the system has a center of inversion located at $(0,0,0)$. The cluster $B$ has the same number of water molecules and is enclosed by a sphere $B$ of the same radius with the center located at $(0,0,-2R_n/3)$. In this minimalist example, each hydrogen and oxygen atom in both clusters is represented by a single $1s$ GTO with exponents $0.122$ and $0.3023$, respectively. These particular values were taken from the cc-pVDZ basis set of the respective atoms~\cite{dunning89}. The subtensor $\mathbf{T}_{\mathrm{sub}}$ considered here consists of all ERI of type $(\mu_A \nu_A|\sigma_B\lambda_B)$, where the subscripts denote the groups to which the AO belong, and these integrals are evaluated exactly using the formula Eq.~(\ref{eq:tsub}). The Cartesian coordinates of the water clusters used in this demonstration are given in the supplementary material.

In order to find the CPD approximation~(\ref{eq:cpd-brutal}) for each system, we carry out the ALS optimization in the Frobenius norm. For each $n$ we initialize the optimization using a small CPD rank of $100$ and increase the rank in relatively small steps of $6n$. At each step, the CPD approximation is fully optimized in consecutive ALS sweeps when three of the factors $A_{\mu r}\,B_{\nu r}\,C_{\sigma r}\,D_{\lambda r}$ in Eq.~(\ref{eq:cpd-brutal}) are fixed while the remaining one is optimized using the conventional least-squares approach. Therefore, a single ALS sweep includes four separate minimization steps. ALS optimization is stopped when the difference in the Frobenius norm error between $10$ consecutive ALS sweeps falls below $10^{-8}$. Note that ALS optimization typically converges very slowly, and thousands of sweeps are typically necessary to reach this demanding convergence goal. The optimization is initialized randomly by drawing numbers for a uniform distribution in the interval $[-1,+1]$. This randomization is carried out both at the beginning of the ALS optimization and during each rank expansion procedure. To speed up the convergence of the ALS procedure, we apply a simple Nesterov acceleration~\cite{nesterov13}, that is, a line search along the current direction of the gradient evaluated as a byproduct of the least-squares step. Note that the ALS procedure with random initialization frequently suffers from numerical instabilities caused by linear dependencies between the CPD factors. To avoid that, we used two remedies. First, Tikhonov regularization~\cite{tikhonov95} (also known as ridge regression) was used by adding a small ($10^{-9}$) constant to the diagonal of the normal matrix in each least-squares step. Second, the first three CPD factors, $A_{\mu r}\,B_{\nu r}\,C_{\sigma r}$, were always constrained to be column-normalized, and only the fourth factor $D_{\lambda r}$ was allowed to carry an arbitrary weight.

In Fig.~\ref{fig:errors} we show the error of the CPD expansion as a function of the CPD rank for the model system described above. As an example, we provide the results for $n=21,27,33$ water molecules within the spheres~$A$~and~$B$. The error of the decomposition initially decreases exponentially (corresponding to a linear curve on the logarithmic scale) and then slightly deviates upward. This behavior is remarkably consistent for all cluster sizes shown in Fig.~\ref{fig:errors}, but also for the remaining values of $n$ considered here. The data shown in Fig.~\ref{fig:errors} is not immediately useful from the point of view of this work, because we want to access the inverse relationship, i.e. the minimum CPD rank for each $n$ that delivers the predefined accuracy $\epsilon$. However, such a rank is difficult to calculate directly because the ALS procedure requires one to specify the rank upfront. To find the precise rank that delivers a given accuracy, one would need to increase the rank in extremely small steps, which is not practical. However, we can approximately find the minimum rank with error no larger than the specified $\epsilon$ by exploiting the regular and smooth behavior of the dependence shown in Fig.~\ref{fig:errors}. To this end, we perform an interpolation of $\epsilon$ as a function of the CPD rank using the third-order $B$-spline method. Knowing this function analytically, we can numerically invert it and then calculate the minimum rank that leads to an error no larger than the specified constant $\epsilon$ for each size of the cluster $n$. This leads to the data shown in Fig.~\ref{fig:ranks} where we plot the minimum rank determined for three choices of $\epsilon$, namely $10^{-2}$, $10^{-3}$, and $10^{-4}$, as a function of the cluster size~$n$.

\begin{figure}
    \centering
    \hspace{-0.5cm}\includegraphics[width=0.7\linewidth]{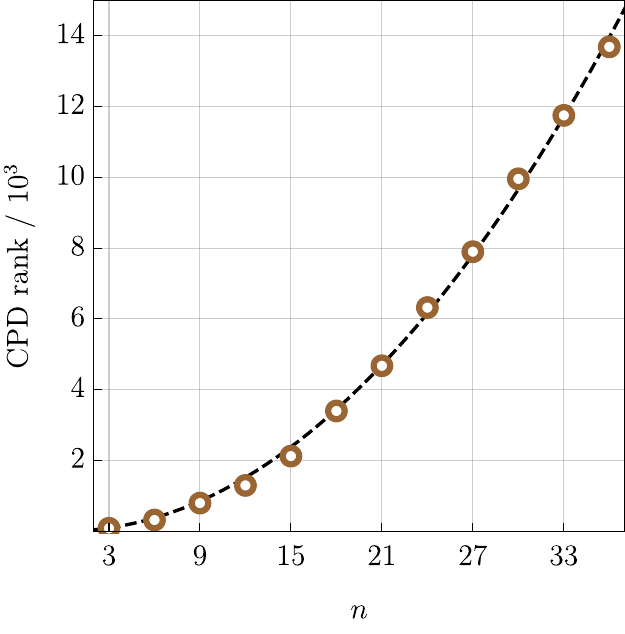}\vspace{0.5cm} \\
    \hspace{-0.5cm}\includegraphics[width=0.7\linewidth]{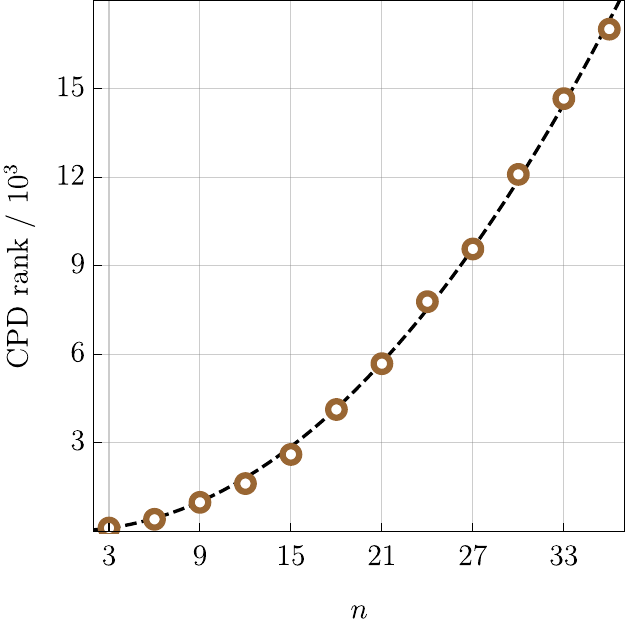}\vspace{0.5cm} \\
    \hspace{-0.5cm}\includegraphics[width=0.7\linewidth]{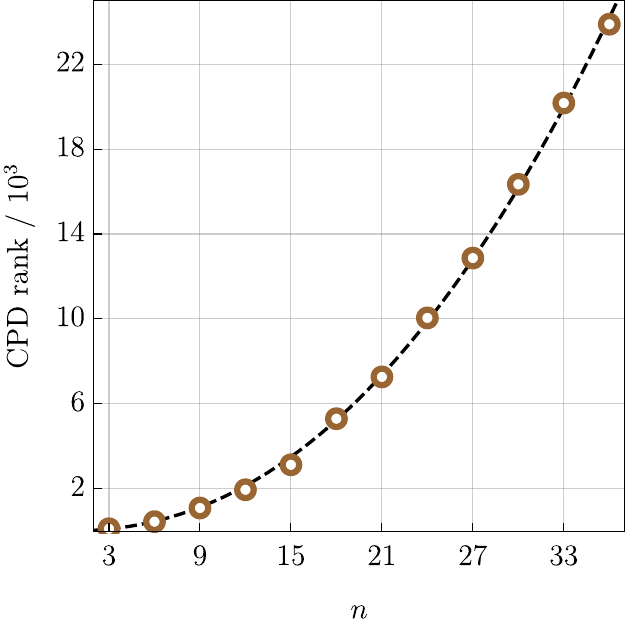}
    \caption{The rank of the CPD approximation necessary to reach the decomposition error $10^{-2}$ (top panel), $10^{-3}$ (middle panel), and $10^{-4}$ (bottom panel) as a function of the number of water molecules $n$ within the spheres~$A$~and~$B$. The dashed black lines are functions in the form $a\cdot n^b$ obtained by least-squares fit to the data points.}
    \label{fig:ranks}
\end{figure}

We can judge the functional dependence of the rank vs. the size of the cluster by fitting the data points with a suitable functional form. For this purpose, we compare the quality of the fit for a set of one-parameter functions: $an$, $an^2$, $an^2\log n$, $an^2/\log n$, as well as two-parameter functions: $an^b$, $an+bn^2$, $an^2+bn^2\log n$, $an^2+bn^2/\log n$, $an+bn^2\log n$, $an+bn^2/\log n$, where $a$, $b$ are adjustable parameters determined by the least-squares procedure. Note that the rank of the CPD approximation should tend to zero in the limit $n\rightarrow 0$ and this constraint has been imposed on each functional form given above. We also initially considered other functional forms allowed, in principle, by Theorem~\ref{th:main}, such as $an^2/\log^m n$ with $m\geq 2$ and their two-parameter generalizations, but all of them provided less accurate representation of the data points compared to the functions listed above.

\begin{table}
\caption{\label{tab:aic}
Comparison of various one- and two-parameter functions in representation of the rank of the CPD approximation (decomposition thresholds: $\epsilon=10^{-2}, 10^{-3}, 10^{-4}$) for the model system of water clusters.
}
\begin{ruledtabular}
\begin{tabular}{lccc}
 function & \multicolumn{3}{c}{Akaike information criterion (AIC)} \\
\cline{2-4}
          & $\epsilon=10^{-2}$ & $\epsilon=10^{-3}$ & $\epsilon=10^{-4}$ \\
\hline
\multicolumn{4}{c}{One-parameter functions} \\
\hline
$an$          & 217.6 & 223.1 & 232.8 \\
$an^2$        & 161.1 & 164.2 & 187.2 \\
$an^2\log n$  & 182.2 & 183.7 & 171.8 \\
$an^2/\log n$ & 189.1 & 195.4 & 209.6 \\
\hline
\multicolumn{4}{c}{Two-parameter functions} \\
\hline
$an^b$               & 162.8 & 163.2 & 164.8 \\
$an+bn^2$            & 161.8 & 161.3 & 161.9 \\
$an^2 + bn^2\log n$  & 162.8 & 163.1 & 163.2 \\
$an^2 + bn^2/\log n$ & 162.3 & 161.9 & 161.2 \\
$an + bn^2\log n$    & 169.7 & 169.6 & 165.9 \\
$an + bn^2/\log n$   & 157.0 & 158.5 & 168.5 \\
\end{tabular}
\end{ruledtabular}
\end{table}

To quantify the ability of each function to represent the underlying data, we consider the Akaike information criterion (AIC)~\cite{akike74}. It is calculated as $2k-2\log L$, where $k$ is the number of parameters in the functional model (one or two in our case) and $L$ is the likelihood function, which is calculated straightforwardly from the least-squares error of the fit. The use of AIC enables us to compare functions with different numbers of adjustable parameters, as the first term ($2k$) effectively penalizes overfitting. Note that a smaller value of AIC signifies a better fit to the data. Moreover, a difference in AIC between two models smaller than $2$ is usually considered statistically insignificant, while a difference greater than $5$ suggests that one of the models is significantly lacking~\cite{burnham02}. Additionally, we note that in a preliminary stage of the analysis, we also tested more complicated functional forms to represent the available data, e.g. rational functions combining the power and logarithmic dependence on $N$. These functional forms had a higher number of adjustable parameters, and some of them provided a slightly better fit than the candidates listed above. However, the AIC calculated for them was actually larger than for the best two-parameter functions considered below, which signals overfitting. Therefore, they were excluded from the final analysis.

In Table~\ref{tab:aic} we gather all proposed fitting functions and the AIC parameter calculated from the fit. We begin our analysis of these results by considering the one-parameter functions. None of the proposed formulas provides a consistently accurate fit of the data across all values of $\epsilon$. In fact, for $\epsilon=10^{-2}$ and $\epsilon=10^{-3}$, a simple quadratic function $an^2$ provides the lowest AIC, i.e. the best representation, but its performance deteriorates significantly for $\epsilon=10^{-4}$. On the other hand, the quadratic-log function $an^2 \log n$ performs significantly worse than $an^2$ for $\epsilon=10^{-2}$ and $\epsilon=10^{-3}$, but is the best performer for $\epsilon=10^{-4}$. The remaining two functions: $an$ and $an^2/\log n$ perform much worse and can be definitively excluded from the list of candidates. In particular, the linear function consistently provides the worst representation, which is also apparent in Fig.~\ref{fig:ranks}, with the difference in AIC compared to the best model of the order of several tens. Overall, the data from Table~\ref{tab:aic} suggests that the growth of the effective rank as a function of the system size is not captured by any of the one-parameter functions in a consistent manner. 

Moving on to the two-parameter functions, the first conclusion is that formula $an + bn^2\log n$ can be safely excluded from the candidate list. For each threshold, it provides an AIC that is larger compared to more accurate models by a statistically significant margin. Of the remaining models, $an + bn^2/\log n$ provides the lowest AIC for $\epsilon=10^{-2}$ and $\epsilon=10^{-3}$, but it deteriorates significantly (by a factor of $10$) for $\epsilon=10^{-4}$. This can be seen as an argument against this candidate, as we expect that a function that correctly captures the underlying physics of the problem behaves consistently for all reasonable values of the threshold. The remaining four functions: $an^b$, $an+bn^2$, $an^2 + bn^2\log n$, and $an^2 + bn^2/\log n$ achieve this requirement because the AIC values fluctuate insignificantly (less than one) as a function of the threshold $\epsilon$. The function $an+bn^2$ provides the lowest AIC on average, with $an^2 + bn^2/\log n$ being the second best. However, the difference in AIC between all four is too small (less than $2$) to draw statistically significant conclusions. The argument against the function $an^b$ could be the fact that there is no theoretical basis to justify such a dependence. It can be seen as a function that captures the growth of the rank as an effective model but nevertheless does not reflect the underlying physics.

In summary, we found that for sufficiently large $n$ the growth of the rank of the CPD approximation is captured either by $n^2$ or by $n^2\log n$ form. There are insufficient data to exclude either of them. This is consistent with the statement of Theorem~\ref{th:main}. Other asymptotic forms, which are not in principle excluded by Theorem~\ref{th:main}, such as $n^2\log^m n$ or $n^2/\log^m n$ with $m\geq 2$ are not supported by numerical results (these forms were not explicitly included in Table~\ref{tab:aic} but were tested at a preliminary stage and removed from the pool of candidates due to consistently much higher values of the AIC). Linear growth of the rank is definitively excluded. Moreover, a two parameter function is required to capture the growth of the effective rank in a consistent manner across different decomposition thresholds. This suggests that the rate of growth may undergo a transition as a function of $n$, for example, increasing purely quadratically for small $n$ but then shifting to $n^2 \log n$ for larger $n$, or \emph{vice versa}.

It is also interesting to compare our numerical findings on rank growth similar data published in the literature. Unfortunately, such systematic studies with increasing CPD rank are sparse, and it is much more common to relate the CPD rank to another quantity, such as the number of orbitals or the size of an auxiliary basis set, allowing them to grow in unison with the system size. Probably the first foundational benchmark of the CPD format was conducted by Benedikt et al.~\cite{benedikt11}. They observed that the effective rank of the decomposition scales roughly between $N^{1.7}$ and $N^{2.6}$ with respect to the number of basis functions, depending on the composition of the atomic basis and the adopted thresholds. The CPD decomposition was also applied more recently in a series of papers by Pierce and collaborators.~\cite{pierce21,pierce22,pierce25,pierce25b} In particular, in Ref.~\cite{pierce26} it has been shown that the CPD rank increases linearly for a set of water clusters and linear alkanes. However, it is unclear whether this is a result of a relatively small size of the systems that were studied (less than $10$ non-hydrogen atoms), such that the rank is swarmed by a large number of short-range interactions between neighboring atoms.

As a final remark, we address the concern that the superlinear growth of the effective rank of the ERI tensor may be mitigated if, instead of minimizing the error norm $|\mathbf{T}-\bar{\mathbf{T}}|$, we used the relative error $\frac{|\mathbf{T}-\bar{\mathbf{T}}|}{|\mathbf{T}|}$ as the optimization target. This approach would take into account that the number of significant integrals naturally increases with the systems size and the same is true for its norm $|\mathbf{T}|$. Let us analyze the consequences of such a modification. For sufficiently large systems, the norm of the ERI tensor becomes proportional to some power of the system size, i.e. $|\mathbf{T}|=c n^\alpha$. The value of $\alpha$ may depend on the composition of the system and its geometry, but certainly $\alpha>0$. Then the requirement that the relative error is below some threshold $\epsilon$, i.e. $\frac{|\mathbf{T}-\bar{\mathbf{T}}|}{|\mathbf{T}|}<\epsilon$, is equivalent to $|\mathbf{T}-\bar{\mathbf{T}|}<\epsilon |\mathbf{T}| = c\epsilon n^\alpha = \epsilon'$. We see that this approach is equivalent to simply performing the same ALS optimization as above, but with increasing the effective threshold $\epsilon'$ in proportion to some power of the system size. We attempted to determine CPD in this way using $\alpha=1/2$ or $\alpha=1$ for the model systems discussed above. Unfortunately, the unintended consequence of this approach is the fact that not only small and numerically insignificant integrals are effectively excluded from the optimization but also large integrals, which are the most important for subsequent evaluation of any quantity, become progressively less accurately represented. As a result, the elementwise maximum absolute error of the decomposition increased rapidly with the system size, instead of staying approximately constant as in the original ALS approach. Therefore, it is unlikely that the CPD determined with a varying threshold $\epsilon'$ can retain sufficient precision to be useful in practice for large systems. 

\section{\label{sec:concl} Conclusions}

In this paper, we have studied the effective rank of the canonical polyadic decomposition applied to the electron repulsion integrals. We have defined a certain model system that represents a broad class of molecules/clusters encountered in practice that can be simultaneously expanded in all spatial directions in a systematic way. We have formulated and mathematically proven a theorem stating that the effective rank of the CPD of ERI is bounded from below by a quantity proportional to $N_{\mathrm{AO}}^2/\log_2^7 N_{\mathrm{AO}}$, where $N_{\mathrm{AO}}$ is the number of atomic orbitals in the molecule, under mild conditions imposed on the decomposition threshold $\epsilon$. As a result, while a subquadratic growth of the CPD rank is not excluded, a linear relationship between the rank and $N_{\mathrm{AO}}$ cannot hold universally. The origin of this superlinear rank explosion has been discussed and traced back to the fundamental inability of a linear-rank global CPD format to represent long-range interactions between orbital products in the regime where the multipole expansion is valid. These findings are numerically confirmed by optimizing the CPD approximation for a subset of ERI for a set of water clusters of increasing size. By analyzing the effective rank given as a function of the number of molecules in the cluster (sufficient to reach a fixed accuracy $\epsilon$), we have found that the rate of growth of the rank is close to quadratic. This is consistent with the statement of the theorem proved in this work. The implications of these findings for the use of the CPD format to represent electron repulsion integrals in quantum chemistry have been analyzed. It has been argued that such a global CPD with quadratic effective rank is likely not competitive with other formats, primarily with THC, especially if one takes into account that robust algorithms are available to determine the latter~\cite{lu15,matthews20,zhang25,zhu26,hillers26}. Nevertheless, it has been pointed out that the CPD format can still be massively useful in simplifying some basic operations in quantum chemistry. The idea is to abandon a global CPD that can be applied in any context, similarly as Cholesky decomposition or THC, and determine a CPD for a subset of integrals that is relevant in a particular application, e.g. that bring a non-negligible contribution to the quantity of interest. In future work, we will shift our focus to designing deterministic and optimization-free algorithms that may fulfill this role.

\section*{Supplementary Material}

See the Supplementary Material for the structures of the water clusters used in the numerical calculations in the paper. The geometries of all water clusters are given in XYZ format, \AA ngstr\"om units. The name of each file is ``slab-n.xyz'' where ``n'' is the number of water molecules in the cluster.

\begin{acknowledgments}
This work was supported by the National Science Centre, Poland, under research project 2024/54/E/ST4/00253. We gratefully acknowledge Poland’s high-performance Infrastructure PLGrid (HPC Centers: ACK Cyfronet AGH, PCSS, CI TASK, WCSS) for providing computer facilities and support within computational grants PLG/2025/018692.
\end{acknowledgments}

\section*{Data Availability Statement}

The data that support the findings of this study are available within the article and its supplementary material.

\appendix

\section{\label{sec:appa} Rank of the overlap matrix}

The objective of this Appendix is to prove Theorem~\ref{th:overlap} that was used in the main text. However, we first need to prove a lemma that the maximum eigenvalue of the overlap matrix, denoted $e_{\mathrm{max}}$ is bounded from above by a constant $K$, i.e. $e_{\mathrm{max}}< K$, that depends just on the composition of the AO basis set for each constituting atom, but does not change as the total number of atoms in the system increases, i.e. $K=\mathcal{O}(1)$.

To show this, we invoke the Gershgorin circle theorem~\cite{Horn_Johnson_1985_ch6} stating that
every eigenvalue $e_r$ of $\mathbf{S}$ lies within at least one disk:
\begin{align}
\label{eq:upper_bound}
    \left| e - S_{\mu\mu} \right| \leq \sum_{\nu\neq\mu} \left| S_{\mu\nu} \right|.
\end{align}
Since $S_{\mu\mu} = 1$, we obtain for the largest eigenvalue:
\begin{align}
    e_{\mathrm{max}} \leq \max_\mu\left( 1 + \sum_{\nu\neq\mu} \left| S_{\mu\nu} \right| \right) = 
    \max_\mu\left( \sum_{\nu} \left| S_{\mu\nu} \right| \right).
\end{align}
All overlap integrals between arbitrary GTOs are bounded from above in the sense that $\left| S_{\mu\nu} \right| \leq Ce^{-\alpha R_{\mu\nu}^2}$, where the constants $C>0$ and $\alpha>0$ depend only on the composition of the basis set, but are independent of the system size. In this definition, we have used the abbreviation $R_{\mu\nu}^2 = | \mathbf{R}_\mu - \mathbf{R}_\nu|^2$. This leads to the upper bound:
\begin{align}
    e_{\mathrm{max}} \leq 
    \max_\mu\left( \sum_{\nu} Ce^{-\alpha R_{\mu\nu}^2} \right).
\end{align}
Let us define subsets of $\nu$ denoted $\mathcal{S}_n = \{ \nu: n \leq \left| \mathbf{R}_\mu - \mathbf{R}_\nu \right| < n + 1 \}$ for $n \in \mathbb{N}$. Then we can rewrite the sum in the following way:
\begin{align}
    \sum_{\nu} e^{-\alpha R_{\mu\nu}^2} = \sum_{n=0}^\infty \sum_{\nu\in \mathcal{S}_n} e^{-\alpha R_{\mu\nu}^2}.
\end{align}
Within each subset $\mathcal{S}_n$ we have $\left| \mathbf{R}_\mu - \mathbf{R}_\nu \right| \geq n $ and hence $e^{-\alpha R_{\mu\nu}^2} \leq e^{-\alpha n^2}$. This leads to:
\begin{align}
    \sum_{\nu} e^{-\alpha R_{\mu\nu}^2} \leq 
    \sum_{n=0}^\infty \sum_{\nu\in \mathcal{S}_n} e^{-\alpha n^2} =
    \sum_{n=0}^\infty e^{-\alpha n^2} \sum_{\nu\in \mathcal{S}_n} =
    \sum_{n=0}^\infty e^{-\alpha n^2} |\mathcal{S}_n|,
\end{align}
where $|\mathcal{S}_n|$ denotes the cardinality (number of elements) of $\mathcal{S}_n$. To find an upper bound for $|\mathcal{S}_n|$ we first recall from Sec.~\ref{sec:system} that the system is enclosed in a sphere with volume $V$ and the linear density of atomic orbitals $\rho$ tends to a constant as the system size increases, $\rho=\mathcal{O}(1)$. Moreover, the quantity $\rho$ is necessarily finite for any system size. This means that if we choose any subsystem enclosed in an arbitrary volume $V'$, there exists a constant $d$ independent of the system size, such that the number of atomic orbitals within $V'$ is smaller than $dV'$. To find the volume $V_n'$ corresponding to each $\mathcal{S}_n$ we note that this set is enclosed by two spheres of radii $n$ and $n+1$. Therefore, $V_n' = \frac{4\pi}{3}(n+1)^3 - \frac{4\pi}{3}n^3 = \frac{4\pi}{3}(3n^2+3n+1)$ and $|\mathcal{S}_n| \leq d V_n' = \frac{4d\pi}{3}(3n^2+3n+1)$. Inserting this into the row sum, we obtain:
\begin{align}
    \sum_{\nu} e^{-\alpha R_{\mu\nu}^2} \leq 
    \frac{4d\pi}{3}\sum_{n=0}^\infty e^{-\alpha n^2}\,(3n^2+3n+1).
\end{align}
The last expression no longer depends on $\mu$ which enables us to write:
\begin{align}
    e_{\mathrm{max}} \leq 
    \frac{4Cd\pi}{3}\sum_{n=0}^\infty e^{-\alpha n^2}\,(3n^2+3n+1).
\end{align}
The remaining infinite sum is convergent, and the whole expression includes only three constants: $d$, $C$, and $\alpha$, which do not change with the system size. This shows that $e_{\mathrm{max}}< K$ and $K=\mathcal{O}(1)$.

With this knowledge, we can shift our focus to proving Theorem~\ref{th:overlap}. First, note that the Frobenius norm of the overlap matrix is equal to the square root of the sum of squares of all its eigenvalues, $e_r$. Without loss of generality, we assume that the eigenvalues are given in an non-increasing order, i.e. $e_1\geq e_2 \geq \ldots \geq e_N$. From Eckart-Young-Mirsky theorem, the best approximation to the overlap matrix of a given rank is obtained by truncating the eigendecomposition by dropping a given number of the lowest eigenvalues. Therefore, the condition that $\mbox{rank}_\epsilon(\mathbf{S})=M$ is equivalent to the relationship $\sum_{r=M+1}^N e_r^2 \leq \epsilon^2$. Let us rewrite the square norm of the overlap matrix exposing both parts:
\begin{align}
    ||S||^2 = \sum_{r=1}^M e_r^2 + \sum_{r=M+1}^N e_r^2.
\end{align}
The second part is bounded from above by $\epsilon^2$. Establishing a useful upper bound for the first sum is equally straightforward by recalling the lemma that $e_{\mathrm{max}}<K$. This allows us to write:
\begin{align}
\label{eq:sum2}
    \sum_{r=1}^M e_r^2 \leq \sum_{r=1}^M e_{\mathrm{max}}^2 = e_{\mathrm{max}}^2\,M < K^2M,
\end{align}
and for the complete square norm we thus have:
\begin{align}
\label{eq:bound1}
    ||S||^2 < K^2\,M + \epsilon^2.
\end{align}
On the other hand, one can also derive a lower bound for $||S||^2$ starting just from the definition and using the fact that AOs are normalized to unity:
\begin{align}
\label{eq:bound2}
    ||S||^2 = \sum_{\mu\nu} S_{\mu\nu}^2 \geq \sum_{\mu} S_{\mu\mu}^2 = 
    \sum_\mu 1 = N.
\end{align}
By combining the bounds from Eqs.~(\ref{eq:bound1})~and~(\ref{eq:bound2}) we obtain the inequality:
\begin{align}
    N < K^2\,M + \epsilon^2,
\end{align}
which can be solved for $M$ giving:
\begin{align}
   M > \frac{N - \epsilon^2}{K^2}.
\end{align}
To finalize the proof, we rewrite this bound in a slightly different form, namely:
\begin{align}
   M > \frac{N - \epsilon^2}{K^2} = 
   \frac{N}{K^2} \left( 1-\frac{\epsilon^2}{N} \right) \geq
   \frac{N}{K^2} \left( 1-\epsilon^2 \right) = c_\epsilon N,
\end{align}
which is exactly the statement from Theorem~\ref{th:overlap}. Note that $c_\epsilon = \frac{1-\epsilon^2}{K^2}$ and therefore depends only on the composition of the basis (through $K$) and the threshold $\epsilon$, but is independent of $N$. Moreover, $c_\epsilon>0$ because $0<\epsilon<1$.

\nocite{*}
\bibliography{cpd}

\end{document}